\documentclass[a4paper,USenglish,modulo]{lipics-v2019} 
\newcommand{\iflong}[1]{#1}
\newcommand{\ifshort}[1]{}

\usepackage[utf8]{inputenc}
\usepackage{cite}
\usepackage{todonotes}
\usepackage[outercaption]{sidecap} 
\usepackage{wrapfig}

\theoremstyle{plain}

\newtheorem{observation}[theorem]{Observation}

{
\theoremstyle{definition}
\newtheorem{myremark}[theorem]{Remark}
}
\makeatletter
\g@addto@macro\bfseries{\boldmath}
\makeatother
\usepackage{mdframed}

\graphicspath{{figures/}}

\newtheorem{fact}[theorem]{Fact}

\usepackage{complexity}
\usepackage{thmtools, thm-restate}

\newcommand{\ee}[1]{\todo[color=green!20]{E: #1}}

\newcommand{\bigoh}{\ensuremath{{\mathcal O}}}
\newcommand{\cH}{\ensuremath{\mathcal H}}
\newcommand{\cG}{\ensuremath{\mathcal G}}

\newcommand{\aV}{V_{\text{add}}}
\newcommand{\aE}{E_{\text{add}}}
\newcommand{\aEpar}{\aE^{H}}

\newcommand{\pH}{\cH^{\times}}	
\newcommand{\pG}{\cG^{\times}}

\newcommand{\dr}[2][]{\ensuremath{\mathcal{#2}
\ifthenelse { \equal {#1} {} }  
{}   
{_{#1}}   
}
}

\author{Eduard Eiben}{Department of Computer Science, Royal Holloway, University of London, Egham, United Kingdom}{eduard.eiben@rhul.ac.uk}{https://orcid.org/0000-0003-2628-3435}{}
\author{Robert Ganian}{Algorithms and Complexity Group, TU Wien, Vienna, Austria}{rganian@ac.tuwien.ac.at}{https://orcid.org/0000-0002-7762-8045}{Robert Ganian acknowledges support by the Austrian Science Fund (FWF, project P31336).}
\author{Thekla Hamm}{Algorithms and Complexity Group, TU Wien, Vienna, Austria}{thamm@ac.tuwien.ac.at}{}{Thekla Hamm acknowledges support by the Austrian Science Fund (FWF, projects P31336 and W1255-N23).}
\author{Fabian Klute}{Deptartment of Information and Computing Sciences, Utrecht University, the Netherlands}%
{f.m.klute@uu.nl}{https://orcid.org/0000-0002-7791-3604}{}
\author{Martin Nöllenburg}{Algorithms and Complexity Group, TU Wien, Vienna, Austria}{noellenburg@ac.tuwien.ac.at}{https://orcid.org/0000-0003-0454-3937}{}

\authorrunning{E.~Eiben, R.~Ganian, T.~Hamm, F.~Klute, M.~Nöllenburg}

\hideLIPIcs


\iflong{\title{Extending Nearly Complete 1-Planar Drawings in Polynomial Time}}
\ifshort{\title{Extending Nearly Complete 1-Planar Drawings in Polynomial Time}}
\titlerunning{A Polynomial-Time Algorithm for Extending Nearly Complete Partial 1-Planar Drawings}

\nolinenumbers

\begin{document}
\maketitle
\begin{abstract}
The problem of extending partial geometric graph representations such as plane graphs has received considerable attention in recent years. In particular, given a graph $G$, a connected subgraph $H$ of $G$ and a drawing $\cH$ of $H$, the extension problem asks whether $\cH$ can be extended into a drawing of $G$ while maintaining some desired property of the drawing (e.g., planarity). 

In their breakthrough result, Angelini et al. [ACM	 TALG 2015] showed that the extension problem is polynomial-time solvable when the aim is to preserve planarity. 
{Very r}ecently {we considered} this problem
for partial 1-planar drawings~[ICALP 2020], which are drawings in the plane that allow each edge to have at most one crossing.
The
most important question identified and left open in that work is whether the problem can be solved in polynomial time when $H$ can be obtained from $G$ by deleting a bounded number of vertices and edges.
In this work, we answer this question positively by providing a constructive polynomial-time decision algorithm.

\end{abstract}

\newcommand{\edg}{embedding graph}

\newpage

\section{Introduction}
	Planarity is a fundamental concept in graph theory and especially in graph drawing, where planar graphs are exactly those graphs that admit a crossing-free node-link drawing in the plane. It is well known that testing whether a graph is planar can be carried out in polynomial time, and in the positive case one can also construct a plane drawing~\cite{defraysseixHowDrawPlanar1990,s-epgg-90}. But what if we are given a more refined question: given a graph $G$ where some subgraph $H$ of $G$ already has a fixed plane drawing $\cH$, is it possible to extend $\cH$ to a full plane drawing of $G$? 	
	The corresponding problem is an example of so-called drawing extension problems, which are motivated (among others) from network visualization applications: there,  important patterns (subgraphs) might be required to have a special layout, or new vertices and edges in a dynamic graph may need to be inserted into an existing (partial) connected drawing which must remain stable in order to preserve the mental map~\cite{mels-lam-95}.

	The problem of extending partial planar drawings was solved thanks to the breakthrough result of Angelini, Di Battista, Frati, Jelinek, Kratochvil, Patrignani
	and Rutter~\cite{adfjkp-tppeg-15}, who provided a linear-time algorithm that answers the above question as well as constructs the desired planar drawing of $G$ (if it exists). 
	Unfortunately, it is often the case that we cannot hope for a plane drawing of $G$ extending $\cH$---either because $G$ itself is not planar, or because the partial drawing $\cH$ cannot be extended to a plane one. A natural way to deal with this situation is to relax the restriction from planarity to a more general class of graphs. 
	In our very recent work~\cite{Aug1P-ICALP}, we investigated the extension problem of partial \emph{1-planar drawings}, one of the most natural and most studied generalizations of planarity~\cite{klm-ab1-17,dlm-sgdbp-19,r-sk-65}. 
	A graph is 1-planar if it admits a drawing in the plane with at most one crossing per edge. Unlike planar graphs, recognizing 1-planar graphs is \NP-complete~\cite{gb-agewce-07,km-mo1h1t-13}, even if the graph is a planar graph plus a single edge~\cite{cm-aepgmcn1h-13}---and hence the extension problem of 1-planar drawings is also \NP-complete~\cite{Aug1P-ICALP}.

	In spite of this initial observation, we showed that the extension problem for 1-planar drawings is polynomial-time solvable when the \emph{edge deletion distance} between $H$ and $G$ is bounded~\cite{Aug1P-ICALP}. However, already in that paper it was pointed out that requiring the edge deletion distance to be bounded is rather restrictive: after all, the deletion of a vertex (including all of its incident edges) from a graph is often considered an atomic operation and yet could have an arbitrarily large impact on the {edge deletion} measure. That is why that article proposed to measure the distance between $H$ and $G$ in terms of the \emph{edge+vertex deletion distance}, i.e., the minimum number of vertex- and edge-deletion operations required to obtain $H$ from $G$. Yet---in spite of providing partial results exploring this notion---the existence of a polynomial-time algorithm for extending partial 1-planar drawings of connected graphs with bounded edge+vertex deletion distance was left as a prominent open question.	
In this paper, we resolve this open question as follows.

\begin{theorem}\label{thm:main}
	Let $\kappa$ be a fixed non-negative integer. Given a graph $G$, a connected subgraph $H$ of $G$ and a 1-planar drawing $\cH$ of $H$ such that $H$ can be obtained from $G$ by a sequence of at most $\kappa$ vertex and edge deletions, it is possible to determine whether $\cH$ can be extended to a 1-planar drawing of $G$ in polynomial time, and if so to compute such an extension.
\end{theorem}

\smallskip
\noindent \textbf{Proof Techniques.}\quad
As the first ingredient for our proof, we use the connectedness of $H$ to obtain a bound on the number of edges in $E(G)\setminus E(H)$ which are pairwise crossing. This allows us to perform exhaustive branching to reduce to the case where all that remains is to insert (a possibly large number of) missing edges incident to at most $\kappa$ vertices (notably those in $V(G)\setminus V(H)$) and where we can assume that these remaining missing edges are pairwise non-crossing.
While this step would seem to represent a significant simplification of the problem, it in fact merely exposes its most challenging part.
This {reduction} step is described in Section~\ref{sec:branching}.

Next, in Section~\ref{sec:baseregions} we analyze the structure of a hypothetical solution in order to partition each cell\footnote{Cells can be viewed as the analogue of faces in 1-planar drawings.} containing at least one missing vertex into 
\emph{base regions}. Intuitively, base regions correspond to a part of a cell which ``belongs'' to a certain missing vertex, in the sense that edges incident to other missing vertices may only interact with a base region in a limited way (but may still be present). 
	We show that every solution has at most $ \bigoh(\kappa^3) $ many base regions, whose boundaries are each determined by the drawing of at most two edges. This allows us to apply a further branching step to identify the boundaries of such base regions.
	
	Third, we show how to subdivide and mark base regions as well as other cells as reserved for drawings of edges that are incident to one of at most two specific added vertices in Section~\ref{sec:inter}.
	The marked cells can be appropriately grouped together into a bounded number of independent subinstances of a restricted problem, where each such subinstance has the crucial property that it only contains missing edges that are incident to its two assigned vertices and must be routed via the subdivided base regions allocated to the subinstance.
%
	
	To complete the proof, Section~\ref{sec:2vtcs} provides an algorithm that can solve the independent subinstances obtained as above. The algorithm expands on 
	the previously developed algorithm for the case of $\kappa=2$~\cite[Section 6]{Aug1P-ICALP} to deal with some added difficulties arising from the fact that the subinstances may geometrically interfere with one another in the plane.


%
%

	\smallskip
	\noindent \textbf{Related Work.}\quad	
	The definition of 1-planarity dates back to Ringel (1965)~\cite{r-sk-65}  and since then the class of 1-planar graphs has been of considerable interest in graph theory, graph drawing and (geometric) graph algorithms, see the recent annotated bibliography on 1-planarity by Kobourov et al.~\cite{klm-ab1-17} collecting 143 references. 
More generally speaking, interest in various classes of beyond-planar graphs (not limited to, but including 1-planar graphs) has steadily been on the rise~\cite{dlm-sgdbp-19,ht-bpg-20} in the last decade.

Our recent work on the extension problem for 1-planar graphs~\cite{Aug1P-ICALP} established the \emph{fixed-parameter tractability}~\cite{DowneyF13,CyganFKLMPPS15} of the problem when parameterized by the edge deletion distance between $H$ and $G$. The proof of that result heavily relied on the fact that the total number of edge crossings introduced by adding the missing edges was upper-bounded by the number of added edges. In particular, this made it possible to define an auxiliary graph $H'$ of bounded treewidth that captured information about the partial drawing $\cH$, whereas the extension problem could then be encoded as a formula in Monadic Second Order Logic over~$H'$. At that point, the problem could be solved by invoking Courcelle's Theorem~\cite{Courcelle90}. 

The same paper used an extension of this idea to solve the extension problem for the more restrictive IC-planar graphs~\cite{albertson08,LiottaM16,Brandenburg18a} with respect to the vertex+edge deletion distance---the key distinction here is that while adding $\kappa$ vertices to an incomplete IC-planar drawing can only create $\kappa$ new crossings, adding just two vertices to an incomplete 1-planar drawing may require an arbitrarily large number of new crossings. As a final result, the paper provided a polynomial-time algorithm that resolved the special case of adding two vertices into a 1-planar drawing; the core of this algorithm relied on dynamic programming and case analysis. A slightly generalized version of this algorithm is also used as a subroutine in the last part of our proof in this paper.

\ifshort{
Other related work also studied extension problems of partial representations (other than drawings) for specific graph classes~\cite{adp-esd-19,Arroyo2019-ext1edge,p-epsd-06,cegl-dgpwpofpa-12,mnr-ecpdg-15,cfglms-dpespg-15,ddf-eupgd-19,br-pclp-17,kkorss-eprpuig-17,kkosv-eprig-17,kkos-eprscg-15,kkkw-eprfgpg-12,cfk-eprcg-13,cdkms-crpgeprh-14,cggkl-pvrep-18}, even in settings where recognition is polynomial-time solvable.
}
\iflong{
	More broadly, we note that the result of Angelini et al.\ is in contrast to other algorithmic extension problems, e.g., on graph coloring of perfect graphs~\cite{DBLP:journals/jgt/KratochvilS97} or 3-edge coloring of cubic bipartite graphs~\cite{f-cepepbg-03}, which are both polynomially tractable but become \NP-complete if partial colorings are specified.
	Again more related to extending partial planar drawings, it is well known by Fáry's Theorem that every planar graph admits a planar straight-line drawing, but testing straight-line extensibility of partial planar straight-line drawings is generally \NP-hard~\cite{p-epsd-06}.
	Polynomial-time algorithms are known for certain special cases, e.g.,
	if the subgraph $H$ is a cycle drawn as a convex polygon, and the straight-line extension must be inside~\cite{cegl-dgpwpofpa-12} or outside~\cite{mnr-ecpdg-15} the polygon.
	Yet, if only the partial drawing is a straight-line drawing and the added edges can be drawn as polylines, Chan et al.~\cite{cfglms-dpespg-15} showed that if a planar extension exists, then there is also one, 
	where all new edges are polylines with at most a linear number of bends. 
	This generalizes a classic result by Pach and Wenger~\cite{DBLP:journals/gc/PachW01} that any $n$-vertex planar graph can be drawn on any set of $n$ points in the plane using polyline edges with $O(n)$ bends.
	Similarly, level-planarity testing takes linear time~\cite{jlm-lptlt-98}, but testing the extensibility of partial level-planar drawings is \NP-complete~\cite{br-pclp-17}.
	Recently, Da Lozzo et al.~\cite{ddf-eupgd-19} studied the extension of partial upward planar drawings for directed graphs, which is  generally \NP-complete, but some special cases admit polynomial-time algorithms.
	On the other end of the planarity spectrum, Arroyo et al.~\cite{adp-esd-19,Arroyo2019-ext1edge} studied drawing extension problems, where the number of crossings per edge is not restricted, yet the drawing must be \emph{simple}, i.e., any pair of edges can intersect in at most one point. 
	They showed that the simple drawing extension problem is \NP-complete~\cite{adp-esd-19}, even if just one edge is to be added~\cite{Arroyo2019-ext1edge}.
	Other related work also studied extensibility problems of partial representations for specific graph classes \cite{kkorss-eprpuig-17,kkosv-eprig-17,kkos-eprscg-15,kkkw-eprfgpg-12,cfk-eprcg-13,cdkms-crpgeprh-14,cggkl-pvrep-18}.
	}

	\medskip


\section{Preliminaries}
\label{sec:prelims}
\noindent \textbf{Graphs and Drawings in the Plane.}
We refer to the standard book by Diestel for basic graph terminology~\cite{Diestel}. For a simple graph $G$, let $V(G)$ be the set of its vertices and $E(G)$ the set of its edges. 

A \emph{drawing} $\cG$ of $G$ in the plane $\mathbb R^2$ is a function that maps each vertex $v \in V(G)$ to a distinct point $\cG(v) \in \mathbb R^2$ and each edge $e=uv \in E(G)$ to a simple open curve $\cG(e) \subset \mathbb R^2$ with endpoints $\cG(u)$ and $\cG(v)$.
For ease of notation we often identify a vertex $v$ and its drawing $\cG(v)$ as well as an edge $e$ and its drawing $\cG(e)$. \ifshort{Throughout the paper we will assume that: (i) no edge passes through a vertex other than its endpoints, (ii) any two edges intersect in at most one point, which is either a common endpoint or a proper \emph{crossing} (i.e., edges cannot touch), and (iii) no three edges cross in a single point.}\iflong{
	We say that a drawing $\cG$ is a \emph{good drawing} (also known as  a \emph{simple topological graph}) if (i) no edge passes through a vertex other than its endpoints, (ii) any two edges intersect in at most one point, which is either a common endpoint or a proper \emph{crossing} (i.e., edges cannot touch), and (iii) no three edges cross in a single point.
	For the rest of this paper we require that a drawing is always good.}
For a drawing $\cG$ of $G$ and $e\in E(G)$, we use $\cG-e$ to denote the drawing of $G-e$ obtained by removing the drawing of $e$ from $\cG$, and for $J\subseteq E(G)$ we define $\cG-J$ analogously.

We say that $\cG$ is \emph{planar} if no two edges $e_1, e_2 \in E(G)$ cross in $\cG$; if the graph $G$ admits a planar drawing, we say that $G$ is planar. A planar drawing $\cG$ subdivides the plane into connected regions called \emph{faces}, where exactly one face, the \emph{outer} (or \emph{external}) face is unbounded.  The \emph{boundary} of a face is the set of edges and vertices whose drawings delimit the face.
Further, $\cG$ induces for each vertex $v \in V(G)$ a cyclic order of its neighbors by using the clockwise order of its incident edges.
This set of cyclic orders is called a \emph{rotation scheme}.
Two planar drawings $\cG_1$ and $\cG_2$ of the same graph $G$ are \emph{equivalent} if they have the same rotation scheme and the same outer face; equivalence classes of planar drawings are also called \emph{embeddings}.
A \emph{plane} graph is a planar graph with a fixed embedding.
For a plane graph, its \emph{dual graph} is defined by introducing a vertex for each face, and connecting two faces by an edge, whenever they are adjacent, i.e.\ share an edge on their boundary.

A drawing $\cG$ is \emph{1-planar} if each edge has at most one crossing and a graph $G$ is \emph{1-planar} if it admits a 1-planar drawing.
Similarly to planar drawings, 1-planar drawings\iflong{ also define a rotation scheme and} subdivide the plane into connected regions, which we call \emph{cells} in order to distinguish them from the faces of a planar drawing.
The \emph{planarization} $G^\times$ of a 1-planar drawing $\cG$ of $G$ is a graph $G^\times$ with $V(G) \subseteq V(G^\times)$ that introduces for each crossing $\gamma$ of $\cG$ a \emph{dummy vertex} $v_{\gamma} \in V(G^\times)$ and that replaces each pair of crossing edges $uv, wx$ in $E(G)$ by the four \emph{half-edges} $uv_{\gamma}, vv_{\gamma}, wv_{\gamma}, xv_{\gamma}$ in $E(G^\times)$, where $\gamma$ is the crossing of $uv$ and $wx$. In addition all crossing-free edges of $E(G)$ belong to $E(G^\times)$.
Obviously, $G^\times$ is planar and the drawing $\pG$ of $G^\times$ corresponds to $\cG$ with the crossings replaced by the dummy vertices.

\medskip

\noindent \textbf{Extending 1-Planar Drawings.}
Given a graph $ G $ and a subgraph $ H $ of $ G $ with a 1-planar drawing $ \cH $ of $ H $, we say that a drawing $ \cG $ of $ G $ is an \emph{extension} of $ \cH $ {(to the graph \(G\))} if the planarization $H^\times$ of $\cH$ and the planarization $\pG$ of $\cG$ restricted to $\cH^\times$ have the same embedding. We can now define our problem of interest.

\begin{mdframed}\label{prob:extension}
	\textsc{1-Planar Drawing Extension}\\
	{\itshape Instance:}  A graph \(G\), a connected subgraph \(H\) of \(G\), and a 1-planar drawing \(\cH\) of \(H\).\\
	{\itshape Task:} Find a 1-planar extension of \(\cH\) to \(G\), or correctly identify that there is none.
\end{mdframed}

A brief discussion about the requirement of $H$ being connected is provided in the Concluding Remarks.

A \emph{solution} of an instance $(G,H,\cH)$ of \textsc{1-Planar Drawing Extension} is a 1-planar drawing $\cG$ of $G$ that is an extension of $\cH$.
We refer to \(\aV := V(G) \setminus V(H)\) as the \emph{added vertices} and to \(\aE := E(G) \setminus E(H)\) as the \emph{added edges}. 
Furthermore, we let 
$\aEpar$
be the set of added edges whose endpoints are already part of the drawing, i.e., 
$\aEpar := \left\lbrace vw \in \aE \mid v, w \in V(H) \right\rbrace$.
It is worth noting that, without loss of generality, we may assume each vertex in $\aV$ to be incident to at least one edge in $\aE$.
 \iflong{Furthermore, it will be useful to assume that $\cH$, $\pH$, $\cG$ and $\pG$ are all drawn atop of each other in the plane, i.e., vertices and edges are drawn in the same coordinates in $\cH$ and $\cG$---this allows us to make statements such as ``a solution $\cG$ draws vertex $v\in \aV$ inside face $f$ of $\pH$''.}

Since \textsc{1-Planar Drawing Extension} is \NP-complete (and remains \NP-complete even if all added edges have at least one endpoint that can be placed freely, i.e., if \(\aEpar = \emptyset\)~\cite{Aug1P-ICALP}), it is natural to strive for efficient algorithms for the case where $\cH$ is nearly a ``complete'' drawing of $G$.
Deletion distance represents a natural and immediate way of quantifying this notion of completeness.
Here, we consider the \emph{vertex+edge deletion distance} $\kappa$ between $H$ and $G$, formalized as $\kappa=|\aV|+|\aEpar|$. We note that the vertex+edge deletion distance is {in general} smaller than the edge deletion distance between $H$ and $G$.

The remainder of the paper is dedicated to a proof of Theorem~\ref{thm:main}, which is achieved by developing a polynomial-time algorithm for \textsc{1-Planar Drawing Extension} when the vertex+edge deletion distance $\kappa$ between $H$ and an $n$-vertex graph $G$ is bounded by a fixed constant.

\section{Initial Branching}
\label{sec:branching}
In this section, we introduce the first ingredient for our proof:
exhaustive branching over the choice in which cell of \(\cH\) every vertex in $ \aV $ will lie in an extension, the drawings of edges in \(\aEpar\), and the drawings of remaining added edges which cross another added edge. In order to perform the last step in polynomial time, we obtain a bound on the number of edges in $\aE$ which are pairwise crossing. This leaves us with a 1-planar drawing $ \cH' $ of some graph $ H' $ with $ H \subseteq H' \subseteq G $ such that $ \aE' = \aE \setminus E(H') $ and for every edge $ uv \in \aE' $ either $ u \in \aV $ or $ v \in \aV $ in each branch. We now provide the details of how all of this is done.

First, note that the number of faces of a planarized 1-planar drawing is linearly bounded in the number of vertices of the original graph~\cite{PachT97}.
Consequently, we can 
exhaustively branch on the choice of cells of \(\cH\) containing the drawings of added vertices in $n^{\bigoh(\kappa)}$ steps.
Recall that once we have decided into which cell of \(\cH\) each added vertex is embedded, the exact position of its embedding is irrelevant in terms of extensibility to \(G\)~\cite{DBLP:journals/gc/PachW01}.
Since $\kappa$ is a fixed constant, this polynomial-time procedure reduces our initial problem to the problem of finding an extension of \(\cH + \mathcal{V}\),
where \(\mathcal{V}\) is an embedding of \(\aV\) into cells of \(\cH\),
to a 1-planar drawing of \(G\).

In the next step we branch over the placement of some edges in $\aE$.
%
To this end, consider the structure of a 1-planar extension of \(\cH + \mathcal{V}\) to a drawing of \(G\) 
and observe that the drawing of an added edge $ e \in \aE $ might:
\begin{enumerate}[(1)]
	\item cross the drawing of at most one different edge in $ \aE $,
	\item cross the drawing of at most one edge of \(H\), or
	\item not cross any edge in $ E(G) $.
\end{enumerate}

We now show that the number of crossings arising from the first case can be bounded by a function of $\kappa$:
\begin{lemma}
		\label{lem:1pXP:few_crossings}
		In any extension of \(\cH + \mathcal{V}\) to a 1-planar drawing of \(G\) there are at most \(|\aEpar| + 3|\aV|^2\) crossings between pairs of edges from \(\aE\).
	\end{lemma}
	\iflong{
		\begin{proof}
			Obviously there are at most \(|\aEpar|\) crossings involving an edge in \(\aEpar\) in an extension of \(\cH + \mathcal{V}\) to a 1-planar drawing of \(G\).
			Similarly{, due to 1-planarity,} there are at most \(|E(G[\aV])| \in \bigoh(|\aV|)\) crossings involving an edge with both endpoints in \(\aV\).
			Now, it suffices to bound the number of crossings between edges that have exactly one endpoint in \(\aV\) by \(2|\aV|^2\).
			For this we consider \(v, w \in \aV\),
			and show that at most two pairs of edges one of which is incident to \(v\) and a vertex in \(V(H)\) and the other to \(w\) and a vertex in \(V(H)\) can have a crossing.
			Then the claim of the lemma follows.
			
			Let \(\mathcal{G}\) be an extension of \(\cH + \mathcal{V}\) to a 1-planar drawing of \(G\).
			Assume, for contradiction, that there are \(v, w \in \aV\) and three pairs \((e_1, f_1), \dotsc, (e_3, f_3)\) of edges whose drawings in \(\mathcal{G}\) cross such that each \(e_i\) is an edge between \(v\) and a vertex in \(V(H)\) and each \(f_i\) is an edge between \(w\) and a vertex in \(V(H)\).
			
			Consider the planarization \(\pG\) of \(\mathcal{G}\) and let \(c_i\) be the vertex introduced for the crossing between \(e_i\) and \(f_i\) for $i=1,2,3$.
			The planarized graph contains all edges of the form \(\{v, c_i\}\), \(\{w, c_i\}\), \(\{x_{e_i}, c_i\}\) and \(\{x_{f_i}, c_i\}\) {where for an edge $q$ we let $x_q$ denote $q\setminus \{v, w\}$, i.e., here $x_q$ will be the endpoint of $q$ in $V(H)$.}
			Since we assume \(H\) to be connected, we can contract the \(x_e\) to a single vertex \(x\) along paths in \(G[H]\) without contracting the edges \(\{v, c_i\}\), \(\{w, c_i\}\) and \(\{x, c_i\}\).
			By assumption \(i = 3\) means we have found a \(K_{3,3}\) minor in a planarized graph, which is a contradiction.
		\end{proof}
	}
	\begin{myremark}
		\label{rem:ce-crossing-bound}
		The claim of Lemma~\ref{lem:1pXP:few_crossings} does not hold if we allow \(H\) to be disconnected. Indeed, Figure~\ref{fig:ce-crossing-bound}
		illustrates how to construct a series of instances with \(|\aV| = 2\) that require \(\bigoh(|V(H)|)\) pairwise crossings between edges in \(\aE\) in \emph{every} solution. 
	\end{myremark}
	\begin{figure}
		\centering
					\includegraphics{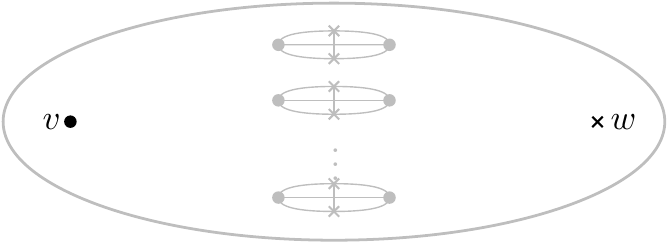}
		\vspace{-0.1cm}
		\caption{Example for Remark~\ref{rem:ce-crossing-bound}.
			\(\cH\) is gray, \(\mathcal{V} = \{v,w\}\), and \(\aE\) connects \(v\) to vertices in \(H\) marked by \(\bullet\) and \(w\) to vertices in \(H\) marked by \(\times\).
			\label{fig:ce-crossing-bound}
		}
		\vspace{-0.4cm}
	\end{figure}
	
\iflong{
Lemma~\ref{lem:1pXP:few_crossings} allows us to apply exhaustive branching to determine which edges will be crossed and how. In particular, we first branch to determine the number $\ell\leq 4\kappa^2$ of edge crossings between pairs of edges from $\aE$, and then which two edges will be involved in the first, second, third,$\ldots$ $\ell$-th crossing. The number of such branches can be upper-bounded by $n^{\bigoh(\kappa^2)}$. Moreover, we note that a pair of pairwise crossing edges can either be drawn inside only one cell, or one of two possible cells (since these edges cannot have any further crossings), and in the latter case we also branch to determine which cell they will be drawn in---after this step, the total number of branches remains upper-bounded by $n^{\bigoh(\kappa^2)}$.

Next, we observe that the drawings of each pair of crossing edges subdivides one cell of \(\cH\) into at most four new cells. Each vertex in $\aV$ that was assigned to the original cell will now belong to one of these four new cells, and we will once again employ branching to determine this. In particular, whenever we place a pair of crossing edges into a cell containing $\ell\leq \kappa$ vertices in $\aV$, we branch to determine which of the four new cells will the added vertices be drawn in; this represents an additional branching factor of at most $4^{\bigoh(\ell)}$ per pair of crossing edges, and hence at most $4^{\bigoh(\kappa^3)}$ in total.
	
    In each branch we can check in polynomial time if the obtained partition induces a partial 1-planar drawing and modify \(\cH\) appropriately if yes; if not, then we discard the corresponding branch.
	Let us call the drawing extension of \(\cH\) up to this point \(\cH_{\text{cross}}\) and the graph that is drawn up to this point \(H_{\text{cross}}\).
	Recall that, based on our branching procedure, we will explicitly assume that drawings of edges missing in \(\cH_{\text{cross}}\) do not mutually intersect in the possible solutions of this branch---in other words, these can either intersect drawings of an edge in \(E(H)\) or no edge.
	
	Moving on, the number of added edges in \(E(G) \setminus E(H_{\text{cross}})\) with both endpoints in \(\aV\) or both endpoints in \(V(H)\) is easily seen to be bounded by \(\kappa^2\).
	This means that we can once again apply exhaustive branching to determine which edge of \(\cH_\text{cross}\), if any, each such edge crosses; if it does not cross another edge of $\cH$ then we will also branch to determine which cell the edge should be drawn in. Altogether, this induces a branching factor of at most $n^{\bigoh(\kappa^2)}$.	
	As above for crossings between edges in \(\aE\), we then branch on the subsets of \(\aV\) arising from the subdivisions of cells in \(\cH_{\text{cross}}\).
	This time, each newly added edge can subdivide at most two cells into at most four subcells in total.
	In this way we consider $4^{\bigoh(\kappa^3)}$ additional branches in total, and---as before---we check whether each branch induces a 1-planar drawing and discard those which do not.
	
	Once again, each branch induces a choice which determines the drawings of the edges involved in such crossings up to extendability.
	
	{Similarly, we connect all connected components which are not yet connected to \(H\) in the partially drawn graph so far by branching on the drawings of at most \(\kappa\) further edges to achieve a drawing of a subgraph of \(G\), each of whose connected components is either connected to \(H\) or all of whose edges are drawn.
	Connected components of the latter type are no longer relevant for extending the respective partial drawing to \(G\) since they contain no endpoints of missing edges and can play no role in separating such endpoints, which is why we omit them from all further considerations.}
	
	A branch for \(\cH_{\text{cross}}\) can be translated into a choice for a 1-plane drawing extension \(\cH'\) of \(\cH_{\text{cross}}\) to a graph \(H_{\text{cross}} \subseteq H' \subseteq G\) in a straightforward way. 
		We formalize this problem and the corresponding statement below. Let a 1-planar extension of $\cH'$ be \emph{untangled} if edges in $E(G)\setminus E(H')$ are mutually non-crossing.
}
\ifshort{
Lemma~\ref{lem:1pXP:few_crossings} allows us to apply exhaustive branching to determine which edges will be crossed, which cell they will be crossed in, and how the previously placed vertices in $\aV$ will be distributed to the new cells created by adding these edges. 
Moving on, the number of added edges with both endpoints in \(\aV\) or both endpoints in \(V(H)\) is easily seen to be bounded by \(\kappa^2\). This means that for each such edge $e$ we can once again apply exhaustive branching to determine which of the edges already present in $\cH$ at this point, if any, $e$ crosses; if $e$ does not cross another edge of $\cH$ then we will also branch to determine which cell $e$ should be drawn in. This, too, requires us to branch on how the vertices in $\aV$ will be distributed to the cells created by these newly placed edges. 
{Similarly, we connect all connected components which are not yet connected to \(H\) in the partially drawn graph so far by branching on the drawings of at most \(\kappa\) further edges to achieve a drawing of a subgraph of \(G\), each of whose connected components is either connected to \(H\) or all of whose edges are drawn.
Connected components of the latter type are no longer relevant for extending the respective partial drawing to \(G\) since they contain no endpoints of missing edges and can play no role in separating such endpoints, which is why we omit them from all further considerations.}

After performing these steps, we are left with a simplified problem which we formally define below.
}
		
\begin{mdframed}\label{prob:extension}
	\textsc{Untangled $\kappa$-Bounded 1-Planar Drawing Extension}\\
	{\itshape Instance:}  A graph \(G\), a
	connected
	subgraph \(H'\) of \(G\) with $V(H')=V(G)$ and with at most $\kappa$ marked vertices such that every edge in $E(G)\setminus E(H')$ has precisely one marked endpoint, and a 1-planar drawing \(\cH'\) of \(H'\).\\
	{\itshape Task:} Find an untangled 1-planar extension of \(\cH'\) to \(G\), or correctly identify that there is none.
\end{mdframed}
Here, a 1-planar extension of $\cH'$ to \(G\) is \emph{untangled} if edges in $E(G)\setminus E(H')$ are mutually non-crossing.
Note that the marked vertices in the problem statement are precisely the vertices in $\aV$. 
\iflong{By applying the branching rules introduced above, we obtain:}\ifshort{By applying the abovementioned branching rules, we obtain:}
\begin{corollary}
\label{cor:reduction}
For every fixed $\kappa$, there is a $n^{\bigoh(\kappa^3)}$-time Turing reduction from \textsc{1-Planar Drawing Extension} restricted to instances of bounded $\kappa$ to \textsc{Untangled $\kappa$-Bounded 1-Planar Drawing Extension}.
\end{corollary}

%
%
%
%
%
%
	
The following observation about the obtained instances of 	\textsc{Untangled $\kappa$-Bounded 1-Planar Drawing Extension} will be useful later on.
\begin{observation}
		\label{obs:H'avertexdeg}
		The degree
		of every vertex in \(\aV\) in the graph \(H'\) lies in \(\bigoh(\kappa^2)\).
	\end{observation}
	\iflong{
	\begin{proof}
		Let \(v \in \aV\).
		All edges incident to \(v\) in \(H'\) are in \(E(H')\setminus E(H)\), and by construction it follows that $|E(H')\setminus E(H)|\in \bigoh(\kappa^2)$.
	\end{proof}
	}
	
	
	\section{Base Regions}
	\label{sec:baseregions}
	
	\newcommand{\nice}{untangled 1-planar}
	
Let us now consider an instance $(G, H', \cH')$ obtained by Corollary~\ref{cor:reduction}. 
	To {simplify} terminology we refer to edges in \(E(G) \setminus E(H')\) as \emph{new edges}.
	\iflong{
	Recall that all new edges are incident to exactly one vertex in \(\aV\) and we want to find an extension $ \cG $ of \(\cH'\) to \(G\) in which the drawings of the new edges only cross edges in \(E(H')\), i.e., an \nice\ extension.}
	
	We will now show that in the planarization of any hypothetical \nice\ extension \(\cG\) of \(\cH'\) to \(G\) we can identify parts of cells of \(\cH'\) which only contain parts of drawings of edges in \(E(\pG) \setminus E(H')\) that are incident to a specific \(v \in \aV\).
	We will call such subsets \emph{base regions} and associate each region to the corresponding \(v \in \aV\).
	Intuitively, base regions of \(v\) determine designated areas \iflong{of the plane }in which new edges incident to $ v $ can start.
	
	\begin{definition}
		\label{def:baseregions}
		A \emph{base region} of some \(v \in \aV\) in \(\cG\) is an inclusion maximal connected subset of a cell of \(\cH'\) containing \(v\), which
		\begin{itemize}
			\item does not contain \(v\);
			\item is bounded by parts of \(\cH'\) and drawings of edges in \(E(\pG) \setminus E(\cH'^{\times})\) which are incident~to~\(v\);
			\item contains the drawing of at least one edge in \(\pG\) which is incident to \(v\); and
			\item contains no drawing of an edge in \(E(G^{\times}) \setminus E(H'^{\times})\) which is incident to some \(w \in \aV \setminus \{v\}\).
		\end{itemize}
	\end{definition}
	\begin{myremark}
	\label{rem:basedifficult}
		An illustration of Definition~\ref{def:baseregions} is provided in Figure~\ref{fig:baseregions}; notice that drawings of new edges of $G$ (not $G^\times$)  with marked endpoints different from $v$ can still intersect the interior of the base region of $v$.
	\end{myremark}
	
	Fixing the boundaries of base regions of all vertices in \(\aV\) to find a hypothetical solution with these base regions determines which edges of \(H'\) can be crossed to draw new edges incident to each region-specific vertex. However, base regions do not give explicit structural restrictions on the drawings of new edges beyond the point at which they cross edges of \(H'\).

\iflong{
	\begin{figure}[btp]
		\centering
		\includegraphics[width=.6\textwidth,page=2]{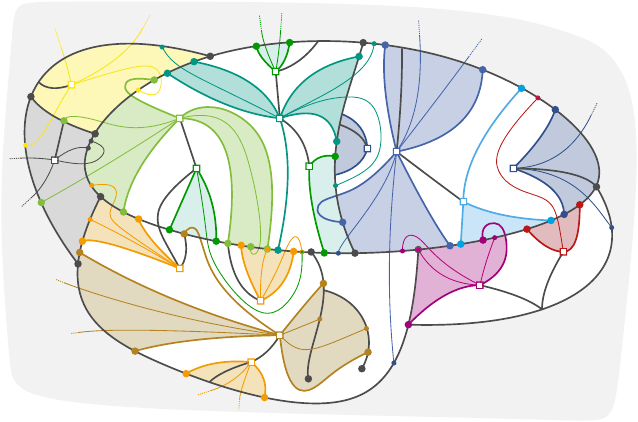}
\caption{Illustration of Definition~\ref{def:baseregions}. 
			\(\cH'\) is the dark-gray graph and the rest of $ \cH' $ is hinted by the light-gray background, %
			the vertices in $ \aV $ are the square marks. %
			Thick colored edges bound the base regions and thin edges are inserted into them. %
			The figure illustrates a few special situations. %
			There are three base regions that each consist only of a single edge (one green, brown, and blue). %
			There is one green base region in which no further edges are added. %
			The yellow and grey vertices have their whole cells as base regions since no other new vertex is inside those cells. %
			The brown vertex has a base region lying completely inside a cell of $ \cH' $ %
			since the part of $ \cH' $ that bounds the base region is 1-connected in $ H' $.%
			\vspace{-0.5cm}}
		\label{fig:baseregions}
	\end{figure}
	}
	

	\begin{myremark}
		\label{rem:baseregions}
		The following basic facts about base regions are easily seen:
		\begin{enumerate}
						\item \label{rem:baseregioncell}\(v \in \aV\) only has base regions in cells containing \(v\) (on their boundary or their interior).
			\item Two base regions intersect only in the boundary of a cell of \(\cH'\).
		\end{enumerate}
	\end{myremark}
	
Our aim for the remainder of this section is to show that we can branch to determine the boundaries of base regions. To this end, it suffices to show that the number of base regions is bounded by a function of \(\kappa\).	
	First, we prove an auxiliary proposition that we then use to show that the number of base regions in each cell of \(\cH'\) lies in \(\bigoh(\kappa)\).
\begin{proposition}
		\label{prop:singlebase}
		In every \nice\ extension of \(\cH'\) to \(G\), each cell of \(\cH'\) that contains an added vertex contains at least one added vertex that has precisely one base region in that cell.
	\end{proposition}
	\iflong{
	\begin{proof}
		Let \(\cG\) be an \nice\ extension of \(\cH'\) to \(G\).
		(If there is none, we are done.)
		Consider a cell \(c\) of \(\cH'\) that contains at least one added vertex.
		Traverse the boundary of \(c\) starting from an arbitrary vertex in counterclockwise direction.
		Note that it might be possible that in this way edges and vertices of \(H'\) may occur two times in this traversal.
		
		Now consider the family \(\mathcal{F}\) of subsets of \(c\) that is bounded by the outermost (in the ordering given by the described traversal) edges in \(E(G^{\times}) \setminus E(H'^{\times})\) which are incident to each \(v \in \aV\) contained in \(c\) and the counterclockwise stretch of the boundary \(c\) between the non-\(v\) endpoints of these edges.
		
		Because we assume that new edges do not mutually cross each other in \(\cG\), \(\mathcal{F}\) of sets is laminar.
		In particular at the lowest level of this laminar family, we find \(v \in \aV\) such that the set \(S_v \in \mathcal{F}\) corresponding to \(v\) does not contain any \(w \in \aV \setminus \{v\}\).
		We show that \(S_v \setminus \{v\}\) is in fact \(v\)'s only base region in \(c\).
		
		We verify that \(S_v \setminus \{v\}\) is a base region of \(v\) in \(c\):
		\begin{itemize}
			\item Obviously \(v \notin S_v \setminus \{v\}\).
			\item By construction of all sets in \(\mathcal{F}\), \(S_v\) is bounded by parts of \(\cH'\) and drawings of edges in \(E(G^{\times}) \setminus E(H'^{\times})\) which are incident to \(v\).
			\item Since for all \(w \in \aV\), \(w\) is separated by the boundary of \(S_v\) from \(S_v\) \(S_v \setminus \{v\}\) cannot contain the drawing of an edge in \(E(G^{\times}) \setminus E(H^{\times})\) which is incident to such a \(w\).
		\end{itemize}
		Finally note that \(S_v \setminus \{v\}\) is connected by construction and because the bounding edges that are incident to \(v\) were chosen to be outermost, inclusion maximal with the previously mentioned criteria.
		
		Moreover, because the edges bounding \(S_v\) were chosen to be outermost, \(S_v\) contains all drawings of edges in \(E(G^{\times}) \setminus E(H'^{\times})\) in \(c\) which are incident to \(v\), and thus \(S_v \setminus \{v\}\) is \(v\)'s only base region in \(c\).
	\end{proof}
	}

\begin{proposition}
		\label{prop:baseregpercell}
		In every \nice\ extension \(\cG\) of \(\cH'\) to \(G\), the total number of base regions  in every cell of \(\cH'\) is at most \(\max(1, 2(\kappa - 1))\).
	\end{proposition}
	\iflong{
	\begin{proof}
		We will prove the proposition by induction on \(\kappa\). Consider first the case that \(\kappa \le 2\).
		If each cell \(c\) of \(\cH'\) that contains at most one vertex in \(\aV\), then obviously \(c\) is the only base region within \(c\) in \(\cG\). Now if a cell \(c\) contains exactly two vertices, \(u,v\), of \(\aV\), then the proposition follows rather straightforwardly from the fact that the new edges do not cross in any \nice\ extension of \(\cH'\) and that the base regions for \(u\) and \(v\) are bounded by new edges with endpoints in \(u\) and \(v\), respectively. 
		
		Now, consider the case that \(\kappa > 2\). 	
		Consider a cell \(c\) of \(\cH'\) which contains some vertices in \(\aV\).
		By Proposition~\ref{prop:singlebase} there is \(v \in \aV\) such that \(v\) has exactly one base region within \(c\) in \(\cG\).
		By induction hypothesis, \(\cG\) restricted to \(G - \{e \in E(G) \setminus E(H') \mid v \in e\}\) has at most \(2(\kappa - 2)\) base regions within \(c\).
		The base region of \(v\) in \(c\) is immediately contained in at most one base region \(b\) within \(c\) in \(\cG\) restricted to \(G - \{e \in E(G) \setminus E(H') \mid v \in e\}\).
		Let \(w \in \aV \setminus \{v\}\) be the vertex for which \(b\) is a base region in \(\cG\) restricted to \(G - \{e \in E(G) \setminus E(H') \mid v \in e\}\).
		Once one takes \(v\) into consideration, \(b\) is subdivided into two base regions for \(w\), each of them bounded by an edge in \(E(G^{\times}) \setminus E(H'^{\times})\) incident to \(w\) bounding \(b\) and an edge in \(E(G^{\times}) \setminus E(H'^{\times})\) incident to \(w\) that is closest possible with respect to a cyclical traversal of the boundary of \(c\) to the boundary of \(v\)'s base region.
		This means the number of base regions within \(c\) in \(\cG\) is at most two larger than the number of base regions within \(c\) in \(\cG\) restricted to \(G - \{e \in E(G) \setminus E(H') \mid v \in e\}\) (\(v\)'s base region and at most 1 additional base region for \(w\)).
	\end{proof}
	}
	\begin{myremark}
		\label{rem:tightbase}
The bound in the proof of Proposition~\ref{prop:baseregpercell} is tight;
		see Figure~\ref{fig:manybases}.
	\end{myremark}
	
	\begin{figure}[btp]
		\centering
		\includegraphics[page=\iflong{2}\ifshort{3}]{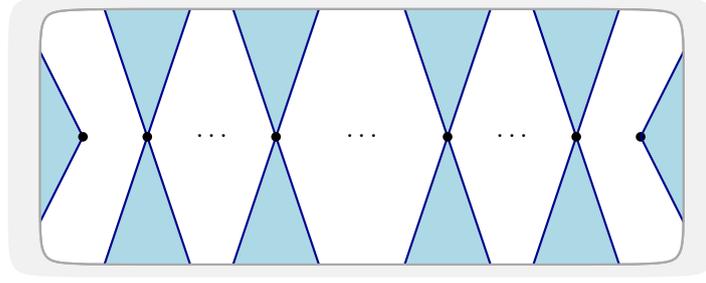}
		\caption{Example for Remark~\ref{rem:tightbase}. \(\cH'\) is gray, \(\aV\) is black and potential base regions are in blue. \vspace{-0.5cm}
			\label{fig:manybases}}
	\end{figure}
	
	In combination with Point~\ref{rem:baseregioncell} of Remark~\ref{rem:baseregions} and the degree bound given in Observation~\ref{obs:H'avertexdeg}, we obtain the following.
	
\begin{lemma}
		\label{lem:baseregions}
		The total number of base regions in any \nice\ extension of \(\cH'\) to \(G\) lies in \(\bigoh(\kappa^3)\).
	\end{lemma}
	\iflong{
	\begin{proof}
		Note that in \(H'\) any \(v \in \aV\) has vertex degree in \(\bigoh(\kappa^2)\) (Observation~\ref{obs:H'avertexdeg}).
		The number of cells of \(\cH'\) containing \(v\) on their boundary or in their interior is at most linear in its vertex degree.
		Hence the number of cells of \(\cH'\) containing \(v\) lies in \(\bigoh(\kappa^2)\).
		Combining this with Remark~\ref{rem:baseregions}.\ref{rem:baseregioncell} and Proposition~\ref{prop:baseregpercell} yields the claim.
	\end{proof}
	}
	
	Observe that every base region is bounded by \(\cH'\) and at most two new edges (since these are each incident to the same vertex in \(\aV\), which is not included in the base region and thus cannot be used for connectivity).
	Hence Lemma~\ref{lem:baseregions} allows us to branch on the drawings of the edges that bound base regions in polynomial time, in a similar fashion as the branching carried out in Section~\ref{sec:branching}.
	In each branch, our aim will be to decide whether the arising 1-planar drawing \(\cH''\) of \(H''\) (where \(H' \subseteq H'' \subseteq G\)) can be extended to an \nice\ drawing of \(G\) with an additional restriction: notably, in the planarization of such an extension, the newly drawn edges immediately incident to \(\aV\) are all drawn in \(\bigoh(\kappa^3)\) distinguished \emph{base cells} of \(\cH''\).
	These base cells are defined as the cells of \(\cH''\) corresponding to the base regions identified in the given branch.

	\section{Interactions between Base Cells}
	\label{sec:inter}
	\newcommand{\vnice}{based untangled 1-planar}
	Let \(\cH''\) be an extension of \(\cH'\) obtained from the previous step described in Section~\ref{sec:baseregions}.
	From now on, we refer to edges in \(E(G) \setminus E(H'')\) as \emph{new edges}.
	Recall that we still want to find an \nice\ extension.
	Additionally, for each \(v \in \aV\), we have identified a set \(\mathfrak{B}_v\) of base cells, {and} we require every new edge incident to \(v\) to start in one such base cell.
			Here, we say that a new edge \(e\) \emph{starts} in or \emph{exits} through the cell in which, for every arbitrarily small $\varepsilon$, the points on the drawing of \(e\) lie at an \(\varepsilon\)-distance from the unique endpoint of \(e\) in \(\aV\).
	We call extensions satisfying this property \emph{\vnice}.
	Furthermore, recall that \(\left\vert \bigcup_{v\in \aV} \mathfrak{B}_v \right\vert \in \bigoh(\kappa^3)\).
	
As noted in Remark~\ref{rem:basedifficult}, edges can cross the boundary of base cells (i.e., ``cross out'' of the base cell they started in) and enter other base cells or cells containing edges from multiple base regions. 
This makes resolving the remaining \textsc{1-planar Extension} problem non-obvious. 
In this and the next section we apply a two-step approach to deal with this issue and complete the proof {of our main result}: first, we subdivide $\cH''$ into parts where only two base cells interact and which can be solved independently (Subsections~\ref{subsec:isointer} and~\ref{subsec:groupinter}), and then we use a dynamic programming algorithm to directly solve each such independent part (Section~\ref{sec:2vtcs}).
\ifshort{Note that the newly added edges could subdivide some of the base cells, technically not making them cells in the extended drawing anymore---however, we always use the term `base cell' to refer to the original base cells for all marking procedures.}
	
	\subsection{Isolating Base Cell Pair Interactions}
	\label{subsec:isointer}
	Our goal here {is to somewhat separate interactions of edges starting in many different base cells.
	To achieve this we aim} to reach a state where each cell is ``assigned'' to at most two base cells{, meaning that only edges starting in these base cells can interact in the respective cell}. 	
	
	We begin with an important definition that will be used throughout this subsection.
	\begin{definition}
		A cell \(c\) of \(\cH''\) is \emph{accessible} from a base cell \(\mathfrak{b} \in \mathfrak{B}_v\) of some \(v \in \aV\), if an edge from \(v\) to a vertex on the boundary of 
		$c$
		can be inserted into \(\cH''\) in a 1-plane way, such that before crossing another edge, it is drawn within \(\mathfrak{b}\), and some part of the edge is drawn within \(c\).
	\end{definition}
	
	\iflong{
	\begin{myremark}
		In particular, a base cell is always accessible from itself.
	\end{myremark}
	}
	
	\begin{observation}
		\label{obs:accessdist}
		A cell of \(\cH''\) that is not a neighbor of a base cell $\mathfrak{b}$ in the dual graph of \(\cH''^{\times}\) is not accessible from $\mathfrak{b}$.
	\end{observation}

	Our next aim will be to bound the number of cells of \(\cH'' \) that are accessible from three different base cells. To do this, we will use the following lemma which is an immediate consequence of the well-known fact that planar graphs have bounded expansion and Point 2 of Lemma 4.3 of previous work by Gajarsk\'y et al.~\cite{GajarskyHOORRVS17}.
	\begin{lemma}[\hspace{-0.001cm}\cite{GajarskyHOORRVS17}]\label{lem:planarBoundedNeighborhood}
		Let $G=(X\cup Y,E)$ be a planar bipartite graph with parts $X$ and $Y$. Then there are at most $\bigoh(|X|)$ distinct subsets $X'\subseteq X$ such that $X'=N(u)$ for some $u\in Y$.
	\end{lemma}
	
	Using Observation~\ref{obs:accessdist} and Lemma~\ref{lem:planarBoundedNeighborhood}, we prove the following.
\begin{proposition}
		\label{prop:accessbound}
		There are at most \(\bigoh(\kappa^3)\) cells of \(\cH''\) which are accessible from three different base cells.
	\end{proposition}
	\iflong{
	\begin{proof}
		Since the number of base cells is already bounded by \(\bigoh(\kappa^3)\), we only need to bound the number of non-base cells accessible from three different base cells.  
		Consider the graph $D$ that is the dual of \(\cH''^{\times}\). For the sake of exposition, we will identify the cells of \(\cH''\) and the vertices of $D$. It follows from Observation~\ref{obs:accessdist} that if a cell \(c\) is accessible from a base cell \(\mathfrak{b}\), then \(c\) and \(\mathfrak{b}\) are adjacent in $D$. Now let \(\mathfrak{B}\) be the set of vertices of \(D\) corresponding to the base cells of \(\cH''\). Since the class of planar graphs is closed under edge deletion, by Lemma~\ref{lem:planarBoundedNeighborhood} there are at most \(\bigoh(|\mathfrak{B}|)\) subsets \(\mathfrak{B}'\)  of \(\mathfrak{B}\) such that $(N_D(c)\cap \mathfrak{B})=\mathfrak{B}'$ for some cell \(c\) of \(\cH''\) that is not a base cell. Clearly for every non-base cell \(c\) which is accessible from three different base cells, we have \(|N_D(c)\cap \mathfrak{B}|\ge 3\).  Now let \(\mathfrak{B}'\) be an arbitrary subset of \(\mathfrak{B}\) of size at least three. Since, \(K_{3,3}\) cannot be a subgraph of a planar graph, it follows from the planarity of $D$ that  
		there are at most $2$ non-base cells \(c\) of \(\cH''\) with \(N_D(c)\cap \mathfrak{B}=\mathfrak{B}'\). Hence, counting also base cells, there are at most \(\bigoh(|\mathfrak{B}|)=\bigoh(\kappa^3)\) cells of \(\cH''\) accessible from three different base cells.	
	\end{proof}
	}
	
	Proposition~\ref{prop:accessbound} allows us to employ a more detailed branching procedure on the structure of a hypothetical solution in cells which are accessible from many (notably, at least three) base cells.
	%
	%
	%
	Our aim is to divide every cell of \(\cH''\) that is accessible from at least three base cells into parts that delimit interactions of pairs of vertices.
	This will then allow us to treat the subcells resulting from this division as cells that are accessible from only two added base cells.
We note that a hypothetical solution will not induce a unique division of a cell into such parts (in contrast to base regions, which are delimited by edges of \cH\ and hence uniquely determined).

	For the remainder of this subsection, let \(c\) be a cell of \(\cH''\) that is accessible from at least three base cells
	and \(\cG\) be a \vnice\ extension of \(\cH''\) to \(G\).
	We proceed as follows:
	First, traverse the boundary of the face of \(\cH''^{\times}\) that corresponds to \(c\), starting from an arbitrary vertex in counterclockwise direction (vertices may appear twice).
	Let the obtained ordering be given by \(v_1, \dotsc, v_{\ell}\).
	Mark each encountered \(v \notin \aV\) with each base cell for which \(v\) is the endpoint of an edge in \(E(\pG)\) which arises from an edge in \(E(G) \setminus E(H'')\) that starts in that base cell (see Figure~\ref{fig:sharedcells} for an example).
	{Note that we mark vertices of a planarized drawing \(\pG\). In particular crossing vertices are also marked.} To avoid confusion, we {also} call attention to the fact that this marking procedure is only defined with respect to a hypothetical solution; keeping that in mind, our next task is to obtain a bound on the number of vertices marked with more than two base cells.
	
\begin{proposition}
		\label{prop:2markbound}
		There are at most \(\bigoh(\kappa^3)\) vertices of \(\cH''\) which are marked with at least three different base cells.
	\end{proposition}
	\iflong{
	\begin{proof}
		By the untangledness assumption, the drawing obtained from \(\cG\) by removing edges in \(\cH''\) is planar.
		The resulting planar graph can be modified by inserting vertices for each base cell, and connecting edges, that start in that base cell to their marked endpoint through this vertex.
		This is easily verified to maintain planarity.
		We call the resulting graph \(G'\) and let \(\mathfrak{B}\) be the newly inserted vertices for base cells.
		Since the class of planar graphs is closed under removal of edges, Lemma~\ref{lem:planarBoundedNeighborhood} implies that for \(v \in V(G') \cap V(H)\), there are at most \(|\mathfrak{B}|\) distinct subsets \(X' \subseteq \mathfrak{B}\) such that \(X' = N_{G'}(v)\).
		Clearly for every vertex in \(V(G') \cap V(H)\) which is marked with three different base cells, we have \(|N_{G'}(v)\cap \mathfrak{B}|\ge 3\).
		Now let \(\mathfrak{B}'\) be an arbitrary subset of \(\mathfrak{B}\) of size at least three.
		Since, \(K_{3,3}\) cannot be a subgraph of a planar graph, it follows from the planarity of $G'$ that  
		there are at most $2$ vertices \(v\) of \(\cH\) with \(N_{G'}(v)\cap \mathfrak{B}=\mathfrak{B}'\).
		Hence, counting also base cells, there are at most \(\bigoh(|\mathfrak{B}|)=\bigoh(\kappa^3)\) vertices marked with three different base cells.
	\end{proof}
	}
	
	Proposition~\ref{prop:2markbound} allows us to branch on the drawings of all missing edges incident to vertices which are marked with three {more} different base cells, and insert them into \(\cH''\). Note that this operation could subdivide some cells which are not base cells; whenever that happens, we recompute the accessibility of the new cells, and we observe that the bound given in Proposition~\ref{prop:accessbound} still applies. On the other hand, the newly added edges could subdivide some of the base cells, technically not making them cells in the extended drawing anymore---in this case, we still use the term `base cell' to refer to the original base cells for all further marking and labeling procedures.

\begin{figure}
  \begin{minipage}[c]{0.6\textwidth}
		\includegraphics[page=3]{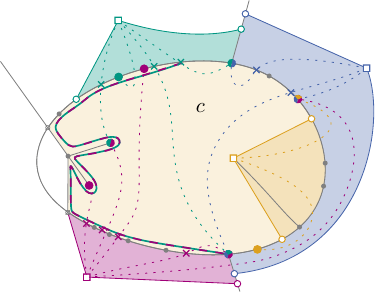}
		\vspace{-0.4cm}
  \end{minipage}\hfill
  \begin{minipage}[c]{0.4\textwidth}
\caption{Illustration of the marking of one cell $ c $ in $ \cH'' $. Square vertices are in $ \aV $ and their base {cells} are marked in the corresponding colors. Filled disks are their neighbors, to which the edges are not yet in \(\cH''\) on the boundary of $ c $.
		Crosses and filled disks receive at least one marker.
		The green-purple curve marks a stretch of the two corresponding base {cells}. Note that the dashed lines represent a possible 1-planar drawing; not all edges are drawn in $ c $.}
				\vspace{-0.4cm}
		\label{fig:sharedcells}
  \end{minipage}
\end{figure}

	After adding the above edges in a branch, we can assume that every vertex is marked by at most two base cells.
	\iflong{	Denote the set of markers for \(v\) by \(\mu(v)\).
	We now identify the following forbidden substructure:

\begin{proposition}
		\label{prop:forbstruct}
		Let \(c\) be a cell of \(\cH''\). There are \emph{no}
		base cells \(\mathfrak{b}_1,\mathfrak{b}_2,\mathfrak{b}_3\) and
		vertices \(u_1,u_2,u_3\), \(v_1,v_2,v_3\) and \(w_1,w_2,w_3\) on the planarized boundary of \(c\),
		whose first occurrence along an arbitrary traversal of the boundary of \(c\) is in this order,
		where possibly \(x_1 = x_2\) or \(x_2 = x_3\) for some \(x \in \{u,v,w\}\),
		such that for all \(i \in \{1,2,3\}\) and \(x \in \{u,v,w\}\), \(\mathfrak{b}_i \in \mu(x_i)\).
	\end{proposition}
	\begin{proof}
		Because \(u_1, u_2, u_3, v_1, v_2, v_3, w_1, w_2, w_3\) occur for the first time in this order along some traversal of the boundary of \(c\) one can connect these vertices by a path in a way that is plane together with \(\cH''^\times\).
		Moreover, \(\mathfrak{b}_1\), \(\mathfrak{b}_2\), \(\mathfrak{b}_3\) can be represented as three distinct	vertices.
		
		By the construction of \(\mu\)  from \(\cH''\) the situation excluded in the statement of the proposition induces the graph that is given by three vertices \(b_1,b_2,b_3\), the disjoint path
		\(u_1, u_2, u_3, v_1, v_2, v_3, w_1, w_2, w_3\)
		where potentially for each  \(x \in \{u,v,w\}\), \(x_1 = x_2\) or \(x_2 = x_3\),
		together with the edges \(\{b_i,x_i\}\) for \(i \in \{1,2,3\}\) and \(x \in \{u,v,w\}\).
		Using the Jordan curve theorem, one can argue that this minor cannot be drawn without the drawings of two edges of the form \(\{b_i,x_i\}\) for \(i \in \{1,2,3\}\) and \(x \in \{u,v,w\}\) intersecting or one of these edges intersecting two edges of the path.
		This contradicts the properties of \(\cG\).
	\end{proof}
	}
	\iflong{One can divide the boundary of \(c\) into connected parts, on which only a specific pair of added vertices appears in the markers. 
	For base cells \(\mathfrak{b},\mathfrak{c}\), we refer to such parts as \emph{\(\mathfrak{b},\mathfrak{c}\)-stretches} in $c$,
	and just \emph{stretches} when we do not want to specify the related base cell pair.
	Such a division is not unique.
	
	We will consider a division which is obtained in a straightforward and greedy way:
	Go through \(v_1, \dotsc, v_\ell\)
	in one of two states: \textit{new stretch}, which is also the initial state, and \textit{running stretch}, which receives the base cells \(\mathfrak{b},\mathfrak{c}\) and a \(\mathfrak{b},\mathfrak{c}\)-stretch as parameters.\\
	If in the state \textit{new stretch} proceed until two different base cells \(\mathfrak{b},\mathfrak{c}\) are encountered in the markers since the beginning of the current \textit{new stretch} state.
	When this happens start a new \(\mathfrak{b},\mathfrak{c}\)-stretch which initially contains every vertex encountered since the beginning of the current \textit{new stretch} state.
	Then switch into \textit{running stretch} state with parameters \(\mathfrak{b},\mathfrak{c}\) and the started stretch.\\
	If in the state \textit{running stretch} with parameters \(\mathfrak{b},\mathfrak{c}\) and a \(\mathfrak{b},\mathfrak{c}\)-stretch to add the encountered vertices to the \(\mathfrak{b},\mathfrak{c}\)-stretch until a vertex contains some base cell \(\mathfrak{d} \notin \{\mathfrak{b},\mathfrak{c}\}\) in its marker (note that we also check this for the vertex in which the current \textit{running stretch} stage starts).
	When this happens switch into \textit{new stretch state}.\\
	When reaching \(v_\ell\) check if the first constructed stretch can be combined with the last stretch under construction (no matter if the last state is \textit{new stretch} or \textit{running stretch}), and merge them if possible.
	Observe that if this is not possible, there is a marker that prevents this, which means we are never in the situation of ending at \(v_\ell\) in the \textit{starting stretch} stage without being able to merge apart from the pathological case in which \(\mu\) maps only to a single base cell.
	This case can be neglected as then no further subdivision of \(c\) is necessary for us to be able to achieve the hypothetical solution \(\cG\).
	
	
	\begin{lemma}
		\label{lem:stretchbound}
		Let \(c\) be a cell of \(\cH''\) and \(\mathfrak{b}\) be a base cell.
		After the described procedure for the cell \(c\), there are at most \(\bigoh(\kappa^3)\) \(\mathfrak{b},\cdot\)-stretches where \(\cdot\) is a placeholder for arbitrary base cells other than \(\mathfrak{b}\).
	\end{lemma}
	\begin{proof}
		Let \(z\) denote the number of base cells.
		Recall that by Lemma~\ref{lem:baseregions} the number of base cells lies in \(\bigoh(\kappa^3)\).
		
		By construction all stretches are inclusion maximal.
		If for all base cells \(\mathfrak{c} \neq \mathfrak{b}\) there are at most two \(\mathfrak{b},\mathfrak{c}\)-stretches,
		there are obviously at most \(\max\{1, 2(z -1)\}\) \(\mathfrak{b},\cdot\)-stretches.
		
		Otherwise by Proposition~\ref{prop:forbstruct} there is exactly one \(\mathfrak{c} \neq \mathfrak{b}\) for which there are more than two \(\mathfrak{b},\mathfrak{c}\)-stretches.
		Note that even stronger, by  Proposition~\ref{prop:forbstruct} then for each \(\mathfrak{d} \notin \{\mathfrak{b},\mathfrak{c}\}\) there are at most two \(\mathfrak{d},\cdot\)-stretches.
		Because of inclusion maximality, the \(\mathfrak{b},\mathfrak{c}\)-stretches have to be separated by some \(\mathfrak{d},\mathfrak{d}'\)-stretches where \(\{\mathfrak{d},\mathfrak{d}'\} \neq \{\mathfrak{b},\mathfrak{c}\}\).
		Let \(\sigma_\mathfrak{b}\) be the number of these stretches, for which \(\mathfrak{b} \in \{\mathfrak{d},\mathfrak{d}'\}\), \(\sigma_\mathfrak{c}\) be the number of these stretches, for which \(\mathfrak{c} \in \{\mathfrak{d},\mathfrak{d}'\}\), and \(\sigma_\cdot\) be the number of these stretches, for which \(\mathfrak{b},\mathfrak{c} \notin \{\mathfrak{d},\mathfrak{d}'\}\){.}
		Then the number \(\sigma\) of \(\mathfrak{b},\cdot\)-stretches is upper-bounded by
		\(\sigma = 2\sigma_\mathfrak{b} + \sigma_\mathfrak{c} + \sigma_\cdot\).
		Also from our arguments above we know \(\sigma_\mathfrak{b} \leq 2(z-2)\), \(\sigma_\mathfrak{c} \leq 2(z-2) - \sigma_\mathfrak{b}\) and \(\sigma_\cdot \leq z - 2 - \sigma_\mathfrak{b} - \sigma_\mathfrak{c}\).
		Clearly \(\sigma\) is maximized for \(\sigma_\mathfrak{b} = 2(z - 2)\) and achieves a maximum value of \(4(z - 2)\).
	\end{proof}
	
	The above lemma allows us to finally identify a set of edges that we can later branch on to reach a state where each cell can be ``assigned'' to at most two designated base cells:
	}\ifshort{One can divide the boundary of \(c\) into connected parts, on which only a specific pair of added vertices appears in the markers.
	It is then possible to show that the number of such connected parts is bounded, 
	which allows us to finally identify a set of edges that we can later branch on to reach a state where each cell can be ``assigned'' to at most two designated base cells.	
}

\begin{lemma}
\label{lem:stretchBranching}
	Let \(c\) be a cell of \(\cH''\).
	There exist a set $F\subseteq (E(G)\setminus E(H''))$ of at most \(\bigoh(\kappa^{6})\) new edges such that if \(\cH_F''\) denotes the restriction of \(\cG\) to $H+F$, then for every cell \(c'\subseteq c\) of \(\cH_F''\) there exist at most two base cells \(\mathfrak{b}^{c'}_1, \mathfrak{b}^{c'}_2\in \bigcup_{v\in \aV}\mathfrak{B}_v\) such that all new edges that intersect \(c'\) in \(\cG\) start either in \(\mathfrak{b}^{c'}_1\) or in \(\mathfrak{b}^{c'}_2\).
\end{lemma}

\iflong{
\begin{proof}
	Given the greedy construction of the stretches for the cell \(c\) as described above,  
	now consider an edge $e$ which intersects the cell $c$ in
	$\cG$.
	{The edge} $e$ must start in some base cell \(\mathfrak{b}\); either \(c\subseteq\mathfrak{b}\), in which case $e$ is drawn from some vertex $v\in \aV$ to a vertex or edge on a \(\mathfrak{b},\cdot\)-stretch on the boundary of $c$, or $c$ is disjoint from $\mathfrak{b}$, in which case $e$ is drawn from an edge on some \(\mathfrak{b},\mathfrak{d}\)-stretch in $c$ to a vertex on some (possibly different) \(\mathfrak{b},\mathfrak{d}'\)-stretch in $c$.
	However, since \(\cG\) is \vnice\ parts of edges drawn into \(c\) do not intersect.
	This implies that the drawings of the at most two outermost edges from each \(\mathfrak{b},\mathfrak{d}\)-stretch to each \(\mathfrak{b},\mathfrak{d}'\)-stretch (where \(\mathfrak{d} = \mathfrak{d}'\) is possible) or to \(v\) itself delimit a subset of \(c\) into which all parts of edges starting in \(\mathfrak{b}\) connecting these stretches (or this stretch and \(v\)) have to be drawn. We will let the set of edges $F$ be precisely the set of the at most two outermost edges from each \(\mathfrak{b},\mathfrak{d}\)-stretch to each \(\mathfrak{b},\mathfrak{d}'\)-stretch. 
	
	Let us first bound the size of $F$. By Lemma~\ref{lem:stretchbound}, there are at most \(\bigoh(\kappa^6)\) stretches in \(c\). Moreover, the new edges in \(\cG\) do not intersect and the stretches of \(c\) are connected subsets of the boundary of \(c\). Hence, the bound on $F$ follows from the planarity of the graph whose vertex set is the set of stretches (+ possibly a vertex $v\in \aV$ {when} \(c\) is a subset of a base cell for $v$) and edge set is the set of pairs such that there is a new edge $e\in F$ connecting the two elements of the vertex set.
	
	Now let us consider a cell \(c'\subseteq c\) that results from including the drawings of edges in $F$ in $\cH''$ whose boundary intersects at least one new edge $e\notin F$. 
	Let \(\mathfrak{b}\) be the base cell in which $e$ starts and let $v$ be its endpoint in $e\cap \aV$. 
	Let us first assume that \(c'\) is not a subset of \(\mathfrak{b}\), then 
	the part of \(e\) that intersects \(c\) is drawn from an edge on some \(\mathfrak{b},\mathfrak{d}\)-stretch in $c$ to a vertex on some (possibly different) \(\mathfrak{b},\mathfrak{d}'\)-stretch in $c$. Since $e\notin F$, there exist two edges $f_1,f_2\in F$ from a vertex/an edge on the \(\mathfrak{b},\mathfrak{d}\)-stretch to a vertex/ an edge on the \(\mathfrak{b},\mathfrak{d}'\)-stretch. Since, $e,f_1,f_2$ do not intersect and $f_1,f_2$ and the two outermost such edges, $c'$ is bounded by a part of this \(\mathfrak{b},\mathfrak{d}\)-stretch, this \(\mathfrak{b},\mathfrak{d}'\)-stretch and edges $f_1,f_2$. It is clear that any edge intersects \(c'\) has to {start} in $\mathfrak{b}, \mathfrak{d}$, or \(\mathfrak{d}'\). If \(\mathfrak{d} = \mathfrak{d}'\), we are done. If  \(\mathfrak{d} \neq \mathfrak{d}'\), we claim that only edges that start in $\mathfrak{b}$ intersect $c'$. Lef $f$ be an edge that starts in $\mathfrak{d}$ and intersects $c'$. $c'$ is incident to only two stretches  \(\mathfrak{b},\mathfrak{d}'\) and  \(\mathfrak{b},\mathfrak{d}\). Furthermore, on \(\mathfrak{b},\mathfrak{d}\) can contain both the intersection point of $f$ with an edge of $H''$ and the endpoint of $f$. Hence $f$ is an edge from \(\mathfrak{b},\mathfrak{d}\) to \(\mathfrak{b},\mathfrak{d}\). But $F$ contains outermost such edge and in particular $f$ has to be separated from $c'$ by $F$ and \(\mathfrak{b},\mathfrak{d}\). Using analogous argument, we obtain that $c'$ does not intersect any edge that starts in  $\mathfrak{d}'$. 
	Finally, if $c'$ is a subset of $\mathfrak{b}$, then we can use similar argument to argue that all edges that start in a different base cell than $\mathfrak{b}$ are separated from $c'$ by an edge in $F$. 
\end{proof}
}

\ifshort{
We now recall that there are at most $\bigoh(\kappa^3)$ cells accessible from at least three base cells, and for each such cell $c$ we will branch to determine a set $F_c$ of at most $\bigoh(\kappa^6)$ edges. We will proceed by assuming that the set $F_c$ is precisely the set of edges obtained by applying Lemma~\ref{lem:stretchBranching} on a hypothetical solution $\cG$. After accounting for some minor technicalities, this gives rise to a branching factor of $n^{\bigoh(\kappa^9)}$. 

 As a consequence, we will proceed under the assumption that all cells have already been marked by {at} most two base cells and every edge that intersects a cell starts in one of the two base cells in the marked set for the cell. However, it is still not possible to cleanly ``split'' an instance into subinstances that only consist of $2$ base cells: the remaining issue is that a vertex $w$ on the boundary between cells assigned to different base cells may still be accessed from multiple cells. The number of times such a situation may occur is \emph{not} bounded by a function of $\kappa$, and hence a simple branching will not suffice here; the next subsection is dedicated to resolving this obstacle.
}
\iflong{
We now recall that there are at most $\bigoh(\kappa^3)$ cells accessible from at least three base cells, and for each such cell $c$ we proceed by branching to determine a set $F_c$ of at most $\bigoh(\kappa^6)$ edges. We will proceed by assuming that the set $F_c$ is precisely the set of edges obtained by applying Lemma~\ref{lem:stretchBranching} on a hypothetical solution $\cG$. Moreover, the drawing of each new edge added in this manner is uniquely determined (up to extendability to $G$) by which edge it crosses (if any); hence, in order to determine the drawing of these edges it suffices to branch to determine their crossings, and we proceed by doing so. 

Each such drawing splits $c$ into $|F_c|+1= \bigoh(\kappa^6)$ cells, and our branching allows us to proceed under the assumption that each of these new cells will be ``accessed'' only by two specific base cells (cf. Lemma~\ref{lem:stretchBranching}). However, at this point we have not yet identified precisely which base cells will be accessing each of the newly created cells---to do so, we can perform an additional branching step to determine which (up to two) base cells will be assigned to each of the newly created cells. This altogether gives us \[
\left(n^{\bigoh(\kappa^6)}\binom{\bigoh(\kappa^3)}{2}^{\bigoh(k^6)}\right)^{\bigoh(\kappa^3)}=n^{\bigoh(\kappa^9)} \] many branches we need to consider.

 As a consequence, we will proceed under the assumption that all cells have already been marked by most two base cells and every edge that intersects a cell starts in one of the two base cells in the marked set for the cell. However, in spite of having completely and exclusively assigned all cells in the incomplete drawing to base cells, it is still not possible to cleanly ``split'' an instance into subinstances that only consist of $2$ base cells: the remaining issue is that a vertex $w$ on the boundary between cells assigned to different base cells may still be accessed from multiple cells, and if an edge needs to be added between $w$ and a vertex $\aV$ it is not clear which base cell such an edge would be drawn in. The number of times such a situation may occur is \emph{not} bounded by a function of $\kappa$, and hence a simple branching will not suffice here; the next subsection is dedicated to resolving this obstacle.
 }

	\subsection{Grouping Interactions}
	\label{subsec:groupinter}
	Let \(\cH'''\) be an extension of \(\cH''\) obtained from the previous step described in Section~\ref{subsec:isointer}.
	\iflong{
	Recall that we can assume that each cell of \(\cH'''\) is only accessible from at most two base cells;
	or it is marked as such.
	We mark all unmarked cells appropriately.
	For a cell \(c\) of \(\cH'''\) we denote the set of its markers by \(\nu(c)\).}
	\ifshort{Recall that we assume that each cell of \(\cH'''\) is accessible from at most two base cells.}
	
	At this point we are roughly in a situation where we could apply our dynamic programming techniques {from} Section~\ref{sec:2vtcs}, namely having to consider only interactions between at most two vertices at a time.
	However we cannot apply such techniques for each cell separately as the cells are not independent.
	More specifically neighbors of vertices in \(\aV\) on the boundary of multiple cells which are accessible to that added vertex can potentially be reached through any of these cells;
	which cell is chosen impacts other edges that can be drawn into that cell, hence possibly forcing them to be drawn into other cells.
	
	\ifshort{We can now employ a similar argument as before Lemma~\ref{lem:stretchBranching} to group cells that are consecutive along the boundary of a base cell. 
For the remainder of this section, let us consider a base cell \(\mathfrak{b}\). 
	}\iflong{
	We can argue similarly as for stretches to group cells which are consecutive along the boundary of a base cell.	
	For this let \(\cG\) be a \vnice\ extension of \(\cH'''\) to \(G\) such that for any cell \(c\) of \(\cH'''\) and any edge in \(e \in E(G) \setminus E(H''')\) part of which is drawn inside \(c\) it holds that the base cell \(\mathfrak{b}\) which \(e\) starts in is an element of \(\nu(c)\).
	For the remainder of this section we fix a base cell \(\mathfrak{b}\).
	Traverse the boundary of \(\mathfrak{b}\) and concurrently track the encountered cells of \(\cH'''\) neighboring \(\mathfrak{b}\), starting from an arbitrary neighboring cell in counterclockwise direction.
	Here we say a cell is encountered, when an edge is encountered that bounds \(\mathfrak{b}\) and the cell in question.
	Note that cells may appear multiple times.
	Let the obtained ordering be given by \(a_1, \dotsc, a_{\ell}\).
	
	\begin{proposition}
		\label{prop:forbstructcells}
		There are \emph{no}
		base cells \(\mathfrak{b}_1,\mathfrak{b}_2\) and
		distinct cells \(c_1,c_2, d_1,d_2\) of \(\cH'''\) neighboring \(\mathfrak{b}\), whose first occurrence along an arbitrary traversal of the boundary of \(\mathfrak{b}\) is in this order,
		such that for all \(i \in \{1,2\}\) and \(x \in \{c,d\}\), \(\nu(x_i) = \{\mathfrak{b}, \mathfrak{b}_i\}\).
	\end{proposition}
	
	\begin{proof}
		It is straightforward to verify that there is no plane drawing of the graph given by the vertex set
		\(\{\mathfrak{b}, c_1, c_2, d_1, d_2\}\) and the edge set
		\(\{\{\mathfrak{b},x_i\}, \{c_i,d_i\}, \{x_1,x_2\} \mid x \in \{c,d\}, i \in \{1,2\}\} \cup \{\{c_2,d_1\}\}\) such that the rotation scheme around \(\mathfrak{b}\) is given by \(c_1, c_2, d_1, d_2\) in this order.
		
		Such a drawing could however be obtained when considering the subgraph of the dual graph of \(\cH'''^\times\) inferred from the situation described in the statement of the proposition.
		For this if \(\mathfrak{b}_1 \notin \{c_1,d_1\}\) one needs to contract \(\mathfrak{b}_1, c_1\) (the appropriate edges exists by Observation~\ref{obs:accessdist}), and similarly if \(\mathfrak{b}_2 \notin \{c_2,d_2\}\) one also needs to contract \(\mathfrak{b}_2, c_2\) (the appropriate edges exists by Observation~\ref{obs:accessdist}.
		Moreover all edges of the dual graph of \(\cH'''^\times\) that correspond to the considered traversal of the neighbors of \(\mathfrak{b}\) and are not incident to can be contracted between \(c_1\) and \(c_2\), between \(c_2\) and \(d_1\), and between \(d_1\) and \(d_2\).
		Note that by these contractions of the edges do not contract the previously obtained edges \(\{c_1,d_1\}\) and \(\{c_2,d_2\}\), as these do not occur between \(c_1\) and \(c_2\), between \(c_2\) and \(d_1\), or between \(d_1\) and \(d_2\).
		
		Either one arrives at the desired path \(c_1, c_2, d_1, d_2\) in this way, or at the star with center \(c_2\) or \(d_1\).
		To resolve the latter case we assume loss of generality that \(c_2\) is the center.
		We can connect \(d_1\) to \(d_2\) that is plane together with the remainder of the considered minor of the dual graph of \(\cH'''^\times\) by tracing the path \(d_1, \mathfrak{b}, d_2\) in an \(\varepsilon\)-distance.
		This could only intersect vertices which lie between \(d_1\) and \(d_2\) in the rotation scheme around \(\mathfrak{b}\).
		Such vertices have however been removed by the contraction of edges.
	\end{proof}
	}One can divide the boundary of \(\mathfrak{b}\) into connected parts, which only separate \(\mathfrak{b}\) from cells marked as accessible from \(\mathfrak{b}\) and some specific second base cell or are not marked as accessible from \(\mathfrak{b}\).
	For a base cell \(\mathfrak{c}\), we refer to such parts as \emph{\(\mathfrak{c}\)-interfaces} (for \(\mathfrak{b}\)),
	and just \emph{interfaces} when we do not want to specify the related base cell. An illustration is provided in Figure~\ref{fig:sharedinterfaces}.
	\begin{figure}
		\begin{minipage}[c]{0.5\textwidth}
			\includegraphics[page=4,scale=.7]{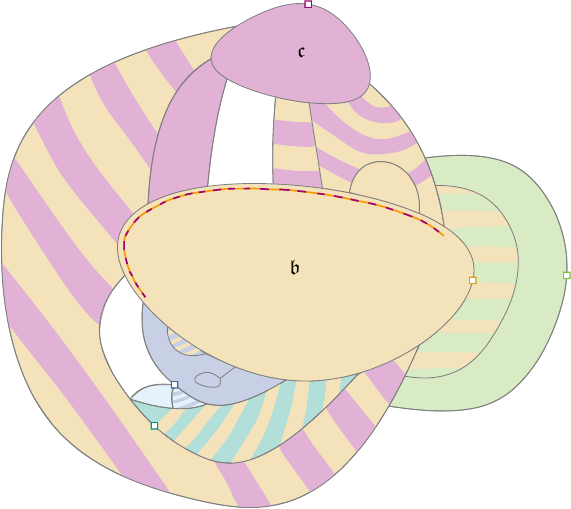}
			\vspace{-0.4cm}
		\end{minipage}\hfill
		\begin{minipage}[c]{0.45\textwidth}
			\caption{Illustration of interfaces. Square vertices are in $ \aV $, and cells are hatched in colors corresponding to a potential accessibility-marking, where every color corresponds to one base cell and a cell is marked as accessible from a base cell if and only if it contains its color.
				The orange-purple curve marks a \(\mathfrak{c}\)-interface from the viewpoint of \(\mathfrak{b}\).}
			\vspace{-0.4cm}
			\label{fig:sharedinterfaces}
		\end{minipage}
	\end{figure}
	
	\iflong{	As for stretches a division of the boundary of \(\mathfrak{b}\) into interfaces is not unique. 
Again we will consider a division which is obtained in a straightforward and greedy way:
	Go through \(a_1, \dotsc, a_\ell\)
	in one of two states: \textit{new interface}, which is also the initial state, and \textit{running interface}, which receives the base cell \(\mathfrak{c}\) and a \(\mathfrak{c}\)-interface as parameters.\\
	
	If in the state \textit{new interface} proceed until encountering a neighboring cell that is marked as accessible from \(\mathfrak{b}\) and second base cell \(\mathfrak{c}\).
	When this happens start a new \(\mathfrak{c}\)-interface which initially contains every vertex encountered since the beginning of the current \textit{new interface} state.
	Then switch into \textit{running interface} state with parameters \(\mathfrak{c}\) and the started interface.\\
	If in the state \textit{running interface} with parameters \(\mathfrak{c}\) and a \(\mathfrak{c}\)-interface add the encountered vertices to the \(\mathfrak{c}\)-interface until a cell is encountered that is marked with \(\mathfrak{b}\) and \(\mathfrak{d}\)  for some base cell \(\mathfrak{d} \notin \{\mathfrak{b},\mathfrak{c}\}\).
	When this happens switch into \textit{new interface state}.\\
	When reaching \(a_\ell\) check if the first constructed interface can be merged with the last interface under construction (no matter if the last state is \textit{new interface} or \textit{running interface}), and do so if possible.
	Observe that if this is not possible, there is a marker that prevents this, which means we are never in the situation of ending at \(a_\ell\) in the \textit{starting interface} stage without being able to merge{,} apart from the pathological case in which \(\nu\) maps only to a single base cell.
	This case can be neglected as then no further subdivision of the cells accessible from \(\mathfrak{b}\) is necessary for us to be able to arrive at the hypothetical solution \(\cG\).}
\ifshort{Given a hypothetical solution $\cG$, we can design a simple greedy procedure 
	that divides accessible cells of $\cH'''$ into interfaces. For the interfaces constructed by this procedure, we can show the following.}

\begin{lemma}
		\label{lem:interfacebound}
		Let \(\mathfrak{c}\) be a base cell.
		After the described procedure, there are at most \(\bigoh(\kappa^3)\) \(\mathfrak{c}\)-interfaces.
	\end{lemma}
	\iflong{	
	\begin{proof}
		Let \(z\) denote the number of base cells.
		Recall that by Lemma~\ref{lem:baseregions} the number of base cells lies in \(\bigoh(\kappa^3)\).
		
		By construction all interfaces are inclusion maximal.
		If for all base cells \(\mathfrak{d} \neq \mathfrak{b}\) there is at most one \(\mathfrak{d}\)-interface,
		there is obviously at most \(1\) \(\mathfrak{c}\)-interface.
		
		Otherwise by Proposition~\ref{prop:forbstructcells} there is exactly one base cell \(\mathfrak{d} \neq \mathfrak{b}\) for which there is more than one inclusion-maximal \(\mathfrak{d}\)-interface.
		Assume this to be the case for \(\mathfrak{c}\).
		Because of inclusion maximality, the \(\mathfrak{c}\)-stretches have to be separated by some \(\mathfrak{d}\)-stretches where \(\mathfrak{d} \notin \{\mathfrak{b}, \mathfrak{c}\}\).
		This means there are at most \(z - 2\) \(\mathfrak{c}\)-interfaces.
	\end{proof}
	}

	We now formalize the subproblem that captures the case of two added vertices which we will obtain at the end of this section. The subproblem is derived on interfaces for every pair of base cells, and will be solved in the next section.

	\newcommand{\subproblem}{\textsc{2-SBCROC}}
	\begin{mdframed}\label{prob:2vsubproblem}
		\textsc{\(2\)-Subdivided Base Cell Routing with Occupied Cells (\subproblem)}\\
		{\itshape Instance:}  A graph \(S\), a connected subgraph \(T\) of \(S\) with \(V(S) = V(T)\) and two vertices \(x\) and \(y\) and a 1-planar drawing \(\mathcal{T}\) of \(T\), with some cells marked as \emph{occupied},
		and there is a {simple} \emph{\(x\)-walk} along boundaries of cells of \(\mathcal{T}\) that have \(x\) on its boundary,
		and there is a {simple} \emph{\(y\)-walk} along boundaries of different cells of \(\mathcal{T}\) that have \(y\) on its boundary.
		Moreover each edge $E(S) \setminus E(T)$ has precisely one endpoint among \(\{x,y\}\).\\
		{\itshape Task:} Find a{n untangled} 1-planar extension of \(\mathcal{T}\) to \(S\) such that no drawing of an edge in \(E(S) \setminus E(T)\) intersects the interior of an occupied cell of \(\mathcal{T}\), and every edge of \(T\) that is crossed by an edge in \(E(S) \setminus E(T)\) incident to \(x\) lies on the {prescribed} \(x\)-walk, and every edge of \(T\) that is crossed by an edge in \(E(S) \setminus E(T)\) incident to \(y\) lies on the {prescribed} \(y\)-walk.
	\end{mdframed}
	
	This is obviously a restriction of \textsc{1-Planar Drawing Extension} for \((S,T,\mathcal{T})\) and we let the relevant terminology (e.g.\ \emph{added edges}) carry over.

\ifshort{
To arrive at appropriate subinstances, Lemma~\ref{lem:interfacebound} and the untangledness of the targeted solution ultimately allows us to branch on delimiting edges for interfaces.
Neighbors of vertices in \(\aV\) where incident edges could still be drawn using a choice of two or more interfaces can be carefully separated and handled either by explicitly branching on the missing drawings of all incident edges, or within a special independent subinstance of \textsc{Untangled $\kappa$-Bounded 1-Planar Drawing Extension} which can be solved using \cite[{Corollary} 17]{Aug1P-ICALP}.
In this way we can show:
}



\begin{lemma}
		\label{lem:reductiontoTwoVertices}
		For every fixed $\kappa$, there is a $n^{\bigoh(\kappa^{28})}$-time Turing reduction 
		from the problem of finding a \vnice\ extension of $\cH'''$ conforming to the current branch	
		 to \subproblem.
	\end{lemma}

\iflong{
\begin{proof}
	
Edges in \(E(G) \setminus E(H''')\) that start in \(\mathfrak{b}\) either are completely drawn within \(\mathfrak{b}\) or cross into some neighboring cell \(c\).
This cell is part of some interface.
As \(\cG\) is \vnice, the drawings of these edges do not intersect.
This implies that the drawings of the at most two outermost edges that start in \(\mathfrak{b}\) and cross into each \(\mathfrak{c}\)-interface delimit a subset of \(\mathfrak{b}\) into which all parts of edges starting in \(\mathfrak{b}\) and crossing into this interface have to be drawn to reach \(\cG\).
What is more, all drawings of edges starting in \(\mathfrak{b}\) lie \emph{completely} in the union {of} one of these subsets with the neighboring cells determined by the corresponding interface.
By Lemma~\ref{lem:interfacebound}, for each pair of base cells, there are at most \({\bigoh(\kappa^3)}\) interfaces. Hence there are \({\bigoh(\kappa^9)}\) interfaces in total and we can enumerate the drawings of edges delimiting these interfaces from the viewpoint of respective base cell boundaries in \(n^{\bigoh(\kappa^9)}\) many branches.

At this point every cell can be assigned to at most two interfaces each from {the} point of view of different base cell. Moreover, if cell $c$ is assigned to two interfaces, one of which is a \(\mathfrak{c}\)-interface of \(\mathfrak{b}\), then it has to be assigned to a \(\mathfrak{b}\)-interface of \(\mathfrak{c}\). Otherwise, we are in an incorrect branch and we can terminate. Now, our goal is to create for each pair of base cell \((\mathfrak{b}, \mathfrak{c})\) an instance that consists of all the cell that are either part of a \(\mathfrak{c}\)-interface of \(\mathfrak{b}\)or of a \(\mathfrak{b}\)-interface of \(\mathfrak{c}\). Note however that this does not completely address the interdependency of cells due to neighbors of added vertices on the boundary of multiple cells which are accessible to that added vertex can potentially be reached through any of these cells.
Instead of having to consider interdependent cells, we now have interdependent interfaces.
In particular, if \(\mathfrak{b}\) is a base cell for \(v \in \aV\) a neighbor which has yet to be connected to \(v\) in \(\cH'''\) can of course also lie on the boundary of multiple interfaces, especially also on interfaces of a different base cell \(\mathfrak{b}'\) for \(c\) and hence can be connected through any of them.

Let us first deal with the case when a vertex is in the boundary of at least $3$ interfaces. We show that for every vertex $v$ there are at most ${\bigoh(\kappa^{27})}$ vertices that are accessible from three different interfaces for some base cell for \(v \in \aV\). We remark, that we do not assume that these interfaces are for two different base cells.
In particular it holds even if all of the interfaces are $\mathfrak{c}$-interfaces for $\mathfrak{b}$ for some fixed  $\mathfrak{c}$ and $\mathfrak{b}$. 

Let us fix three pairwise distinct interfaces $\mathfrak{I}_1$, $\mathfrak{I}_2$, $\mathfrak{I}_3$ and three base cell $\mathfrak{b}_1$, $\mathfrak{b}_2$, and $\mathfrak{b}_3$ for \(v \in \aV\) such that $\mathfrak{I}_1$ is an interface of $\mathfrak{b}_1$, $\mathfrak{I}_2$ an interface for $\mathfrak{b}_2$ and $\mathfrak{I}_3$ is an interface for $\mathfrak{b}_3$. Since there are only $\bigoh(\kappa^3)$ many base cell, there are at most $\kappa^9$ choices for the base cells $\mathfrak{b}_1$, $\mathfrak{b}_2$, and $\mathfrak{b}_3$ and by Lemma~\ref{lem:interfacebound}, there are at most $\bigoh(\kappa^6)$ choices for each of interfaces $\mathfrak{I}_1$, $\mathfrak{I}_2$, and $\mathfrak{I}_3$, respectively. Therefore there are at most ${\bigoh(\kappa^{27})}$ many triples $\mathfrak{I}_1$, $\mathfrak{I}_2$, $\mathfrak{I}_3$ and it suffice to show that for each such choice of interfaces there are at most constantly many vertices in all three of these interfaces. We will actually show that there is at most one such vertex.  

By definition, for $i\in \{1,2,3\}$, the interface $\mathfrak{I}_i$ for \(\mathfrak{b}_i\) is a part of the boundary of \(\mathfrak{b}_i\) in a connected curve, \emph{i.e.,} a subpath of $H'''$. Let $P_i$ be such subpath defined for $\mathfrak{I}_i$. Moreover, from definition of the interfaces it follows that $P_i$ and $P_j$ for $i\neq j$ do not intersect and for every vertex $x$ on the boundary of $\mathfrak{I}_i$, we can draw a edge from an vertex in $P_i$ to $x$ wholly inside $\mathfrak{I}_i$. It follows that if we contract each $P_i$ to a single vertex $p_i$, the complete bipartite graph with bipartition $\{p_1,p_2,p_3\}$ and $\{v\}\cup \{x\mid x\mbox{ is on the boundary of } \mathfrak{I}_i, i\in \{1,2,3\} \}$ is planar. Since planar graph cannot contain $K_{3,3}$ as an subgraph, it follows that there can be at most one vertex that is in all three interfaces.

As there are only $\kappa$ vertices in $\aV$, {i}t follows that there are at most $\bigoh(\kappa^{28})$ new edges that can each access their endpoint, different than the one in $\aV$, from at least three different interfaces. Since at this point we already identified the interfaces, we can easily identify these edges and we can branch on all possible $n^{\bigoh(\kappa^{28})}$ drawings of these edges. 
From now on, we can assume that all vertices are in at most two interfaces.

Now the neighbors of any \(v \in \aV\) can almost all uniquely be assigned to an interface.
The only neighbors of \(v\) for which this is not uniquely possible are those that lie on the shared boundaries of cells \(c_1\) and \(c_2\) that belong to two different interfaces (more than two have been handled by above argument).
By the bound on the pairs of interfaces, there are at most \(\bigoh(\kappa^{18})\) pairs of interfaces that might contain the shared boundary of \(c_1\) and \(c_2\) respectively.
Let \(d_1\) be in a \(\mathfrak{b}_1\)-interface from the viewpoint of \(\mathfrak{c}_1\), and \(d_2\) be in a \(\mathfrak{b}_2\)-interface from the the viewpoint of \(\mathfrak{c}_2\), and without loss of generality \(b_1\) and \(b_2\) are two (not necessarily distinct) base cells for \(v\).
By a similar approach as used for stretches as constructed in Subsection~\ref{subsec:isointer} and interfaces constructed in Subsection~\ref{subsec:groupinter} we can traverse the boundary of the union of all cells touching any fixed interface, and bound the number of connected parts along this boundary intersecting at most one other fixed interface by \(\bigoh(\kappa^9)\).

Again, similarly to the consideration for stretches and interfaces, we can branch on the drawings of missing ``delimiting'' edges in a hypothetical solution incident to \(v\) for each such connected part, which are just the outermost edges (with respect to a traversal of the unions of the interfaces touching \(c_1\) and \(c_2\)) edges incident to \(v\).
These delimiting edges then form a closed curve that separates the space enclosed by this curve from every vertex in \(\aV \setminus \{v\}\).
This enclosed space can then be treated as an instance of \textsc{1-Planar Drawing Extension} with \(|\aV|=1\), which can be solved in polynomial time using \cite[Corollary 17]{Aug1P-ICALP}, and no further edge to \(v\) needs to be drawn to the intersection outside the enclosed space.

All together this branching requires time in \(\bigoh(n^{\kappa^{18}})\).



At this point we assume that for each edge in \(E(G) \setminus E(H''')\) the interface which contains part of this edge in a targeted hypothetical solution is uniquely determined.
	In particular, our interfaces give rise to the following instances of \subproblem:
For every two base cells \(\mathfrak{b}\) and \(\mathfrak{c}\) in \(\cH''\) of \(x\) and \(y\) respectively,  consider the instance \((S,T,\mathcal{T})\) given as follows:
\(S\) is the subgraph of \(G\) which arises from \(G\) by removing all edges in \(E(G) \setminus E(H''')\) that are not incident to \(x\) or \(y\),
and also removing edges in \(E(G) \setminus E(H''')\) which are incident to \(x\) and whose endpoint in \(V(H)\) is not branched to be within a \(\mathfrak{c}\)-interface from the viewpoint of \(\mathfrak{b}\),
and edges in \(E(G) \setminus E(H''')\) which are incident to \(y\) and whose endpoint in \(V(H)\) is not branched to be within a \(\mathfrak{b}\)-interface from the viewpoint of \(\mathfrak{c}\).
By construction \(\cH'''\) is a subgraph of \(G\).
We mark all cells which are not subcells of \(\mathfrak{b}\), \(\mathfrak{c}\) or part of a \(\mathfrak{c}\)-interface from the viewpoint of \(\mathfrak{b}\) or a \(\mathfrak{b}\)-interface from the viewpoint of \(\mathfrak{c}\) as occupied.
What is more, by assumption on the targeted hypothetical solution to our original \textsc{1-Planar Drawing Extension}-instance any edge missing from \(\cH''\) that starts in a base cell of \(\cH''\) connects to a vertex, or crosses an edge on the walk along the boundary of this base cell between the neighbors of \(v\) on that boundary, without \(v\).
What is more any edge missing from \(\cH'''\) that starts in a base cell of \(\cH''\) connects to a vertex, or crosses an edge on the walk along the boundary of one of the subdivisions in \(\cH'''\) of this base cell base cell between the neighbors of \(v\) on that boundary, without \(v\).
These walks for \(x\) and \(y\) respectively are appropriate choices for the \(x\)-walk and the \(y\)-walk.
With these definitions of occupied cells, the \(x\)-walk and the \(y\)-walk we can consider the instance \((S,H''',\cH''')\) of \subproblem.

From the results in this section, we have that a branch for solving \textsc{1-Planar Drawing Extension} for \((G,H,\cH''')\) leads to a solution, whenever every added edge is considered as an added edge in one of these instances, and all the considered instances have solutions.
\end{proof}	
}
	

	\section{1-Plane Routing of Two Vertices}\label{sec:2vtcs}
	In this section we give an algorithm to solve \subproblem\ for an arbitrary instance \((S,T,\mathcal{T})\), which ultimately allows us to prove our main result.
	\iflong{In particular we consider two distinguished vertices, \(x\) and \(y\).
	We refer to the cells of \(\mathcal{T}\) along the \(x\)-walk that contain \(x\) on their boundary as \emph{starting cells} for \(x\), and the cells of \(\mathcal{T}\) along the \(y\)-walk that contain \(y\) on their boundary collectively as \emph{starting cells} for \(y\).
	}

	The idea of the algorithm to solve \subproblem\ is for the most part the same as in \cite[Section 6]{Aug1P-ICALP}.
	In particular our algorithm employs a carefully designed dynamic ``delimit-and-sweep'' approach to iteratively arrive at situations which can be reduced to a network-flow problem to which standard maximum-flow algorithms can be applied.
	However several adaptations can, or have to, be made due to the fact that the instances considered here have a slightly different structure (e.g.\ this algorithm also handles cases, where \(x\) and \(y\) lie on the boundaries of multiple cells) than the ones considered for the special case considered in \cite{Aug1P-ICALP} and are also more general (e.g.\ occupied cells have to be taken into acount).

	\iflong{
	Because of this, and the sake of self-containment, we include a description of the procedure in some detail.
	In particular, we describe the situation which we can reduce to a network flow problem in Section~\ref{subsec:flow}.
	The main dynamic program is described in \ref{subsec:dynprog}; compared to the algorithm given in \cite{Aug1P-ICALP}, in the instanciation of the problem is somewhat different---once again the well-structuredness of a hypothetical solution achieved by the previous steps allows for a somewhat more elegant treatment, but more importantly imposes additional restrictions, that impact other derived subinstances and have to be additionally considered.
	
	\subsection{A Flow Subroutine}
	\label{subsec:flow}
	Our dynamic program is based on a generic network-flow subroutine that allows us to immediately parts of certain derived instances of \textsc{1-Planar Drawing Extension}.
	For this let \((S,T,\mathcal{T})\) be an instance of \subproblem.
	Let \(\lambda\) be a function from the cells of \(\mathcal{T}\) to \(\{x, y,\emptyset\}\), that maps every occupied cell to \(\emptyset\).
	
	We say a 1-planar extension of \(\mathcal{T}\) to \(S\) is \emph{\(\lambda\)-consistent} if whenever the drawing of {any} edge in \(E(S) \setminus E(T)\) which is incident to \(x\) intersects the interior of cell \(c\) of \(\mathcal{T}\) which is not a starting cell of \(x\), then \(\lambda(c) = x\),
	and correspondingly whenever the drawing of any edge in \(E(S) \setminus E(T)\) which is incident to \(y\) intersects the interior of cell \(c\) of \(\mathcal{T}\) which is not a starting cell of \(y\), then \(\lambda(c) = y\) (i.e., $\lambda$ restricts the drawings of which added edges may enter which face).
	{Note that a \(\lambda\)-consistent drawing is always untangled.}
	We show that for a given \(\lambda\) we can either find a \(\lambda\)-consistent solution of \((S,T,\mathcal{T})\) for \subproblem\ or decide that there is none, by constructing an equivalent network flow problem.
	
	\begin{lemma}
		\label{lem:flow}
		Given $\lambda$ as above, it is possible to determine whether there exists a $\lambda$-consistent solution of \((S,T,\mathcal{T})\) for \subproblem\ in polynomial time.
	\end{lemma}
	\begin{proof}
		Consider the maximum flow instance $\theta_1$ constructed as follows.
		We let $\theta_1$ contain a universal sink $t$ and a universal source $s$.
		We add one ``neighbor-vertex'' \(v_w\) for each vertex in $w \in N_{(V(S),E(S) \setminus E(T))}(x)$, and every such neighbor-vertex has a capacity-1 edge to $t$.
		We also add one ``starting-vertex'' \(v_{c_x}\) for every starting cell \(c_x\) of \(x\), and an edge with unlimited capacity from \(s\) to \(v_{c_x}\).
		Moreover, we add one ``cell-vertex'' \(v_c\) for every cell $c$ of $\mathcal{T}'$ which is not a staring cell of \(x\) that $\lambda$ maps to $x$, and add a capacity-1 edge from each such vertex to every neighbor-vertex $v_w$ for which \(w\) lies on the boundary of $c$.
		Finally, we add an edge from every starting-vertex \(v_{c_x}\) to every cell-vertex $v_c$ for which \(\lambda(c) = x\) and set the capacity of this edge to be the number of crossable edges of \(\mathcal{T}'\) that lie on the shared boundary of $c_x$ and \(c\).
		(In particular, some of these edges may have capacity \(0\).)
		The instance $\theta_2$ is then constructed in an analogous fashion, but for $y$.
				
		Assume there is a $\lambda$-consistent extension \(\mathcal{S}\) of \(\mathcal{T}\) to \(S\).
		Each edge from $x$ to an endpoint $v\in N_{(V(S),E(S) \setminus E(T))}(x)$ is either non-crossing (in which case $\theta_1$ models it as a flow from \(s\) through a starting-vertex \(v_{c_x}\) and the neighbor-vertex \(v_v\), and then to $t$),
		or crosses into another cell $c \neq c_x$ (in which case $\theta_1$ models it as a flow from a starting-vertex $v_{c_x}$ to $v_c$, then to $v_v$, and finally to $t$).
		The fact that each such edge in $\mathcal{S}$ must use a separate crossable edge when crossing out of a starting cell \(c_x\) ensures that routing the flow in this way will not exceed the edge capacities of $\theta_1$.
		This construction achieves a flow value of \(|N_{(V(S),E(S) \setminus E(T))}(x)|\).
		The argument for $\theta_2$ is analogous, and hence we obtain that $\theta_1$ and $\theta_2$ allow flows of values \(|N_{(V(S),E(S) \setminus E(T))}(x)|\) and \(|N_{(V(S),E(S) \setminus E(T))}(y)|\) respectively.
				
		Conversely, assume that both $\theta_1$ and $\theta_2$ allow flows of values \(|N_{(V(S),E(S) \setminus E(T))}(x)|\) and \(|N_{(V(S),E(S) \setminus E(T))}(y)|\) respectively.
		Consider such a flow for $\theta_1$ (the procedure for $\theta_2$ will be analogous).
		Clearly, this flow must route capacity \(1\) through each neighbor-vertex $v_v$ to achieve its value.
		If the flow enters $v_v$ directly from a starting-vertex \(v_{c_x}\), then we can add a drawing of the edge from $x$ to $v$ passing only through $c_x$.
		If it instead enters $v$ after passing through a starting-vertex \(v_{c_x}\) \emph{and} a cell-vertex \(v_c\), then we can add a drawing of the edge from $x$ to $v$ starting in $c_x$ and crossing into $c$.
		Since no edge from $y$ enters a cell for with \(\lambda\)-value \(x\), there is a way of choosing which boundaries to $c$ to cross in order to ensure that edges from $x$ will not cross each other (the order of crossings matches the order on which the neighbors $v\in N_{(V(S),E(S) \setminus E(T))}(x)$ appear on the boundary of $c$).
		Applying the same argument for $\theta_2$ results in a $\lambda$-consistent extension $\mathcal{S}$ of \(\mathcal{T}\) to \(S\).
	\end{proof}

	Observe that for a hypothetical 1-planar extension \(\mathcal{S}\) of \(\mathcal{T}\) to \(S\), it is \emph{not necessarily} true that for any restriction \(\mathcal{T}'\) of \(\mathcal{S}\) to some \(T \subseteq T' \subseteq S\), there is a \(\lambda\) such that \(\mathcal{S}\) is a \(\lambda\)-consistent extension of \(\mathcal{T}'\).
	More specifically this is not always the case, because some cells of \(\mathcal{T}'\) might need to be further subdivided, before they can be completely assigned to \(x\) or \(y\) by any \(\lambda\).
	Obviously this is in particular the case for \(\mathcal{T}\) itself. 
	In view of this above discussion our aim will be to branch on the additional drawings of some edges into \(\mathcal{T}\) such that for this extended drawing \(\mathcal{T}'\) we are able to iteratively extend an assignment \(\lambda\) for which we iteratively extend a \(\lambda\)-consistent to an extension \(\mathcal{S}\) to \(S\), by applying Lemma~\ref{lem:flow}.

	\subsection{Dynamic Programming}
	\label{subsec:dynprog}
	Our dynamic program will proceed along the \(x\)-walk and the \(y\)-walk to extend \(\mathcal{T}\) and a partial function \(\lambda\) that maps from some cells of that extension to \(\{x,y\}\), with the aim of applying Lemma~\ref{lem:flow}.
		
	To explicitly distinguish between edges that we no longer or that we still can cross, we mark edges on the \(x\)-walk and on the \(y\)-walk as \emph{uncrossable} whenever they are crossed in \(\mathcal{T}\).
	Additionally, we mark every edge on the shared boundary of a starting cell of \(x\) or \(y\) and an occupied cell of \(\mathcal{T}\) as uncrossable.
	All other edges on the \(x\)-walk and on the \(y\)-walk are marked as \emph{crossable}.
	
	We traverse the \(x\)-walk and the \(y\)-walk in counterclockwise and clockwise direction (from the perspectives of \(x\) and \(y\)) respectively.
	Denote the vertices on the \(x\)-walk in the order described by the traversal as \(x_1, \dotsc, x_{\text{end}_x}\),
	and the vertices on the \(y\)-walk in the order described by the traversal as \(y_1, \dotsc, y_{\text{end}_y}\).
	Note that then vertices that are on the shared boundary of starting cells of \(x\) and \(y\) correspond to some \(x_i\) as well as some \(y_j\).
	Also observe that, because of the ``compatible'' directions of traversals of the \(x\)-walk and the \(y\)-walk, vertices that are on the shared boundary of starting cells of \(x\) and \(y\) occur in the same order in the traversal of the \(x\)-walk and the traversal of the \(y\)-walk, i.e.\
	if \(x_i = y_j\) and \(x_{i'} = y_{j'}\) then \(i \leq i'\) implies that \(j \leq j'\).
	
	\paragraph*{Colored Terminology.}
	Paralleling the the presentation in \cite{Aug1P-ICALP}, we shift to a more colorful terminology which is also helpful for visual representation in figures.
	We associate cells of \(\mathcal{T}\) which have a crossable edge of the \(x\)-walk but no crossable edge of the \(y\)-walk on their boundary with the color \emph{red},
	cells of \(\mathcal{T}\)  have a crossable edge of the \(y\)-walk but no crossable edge of the \(x\)-walk on their boundary with the color \emph{blue},
	and cells which have a crossable edge of the \(x\)-walk as well as a crossable edge of the \(y\)-walk on their boundary with the color \emph{red} with the color \emph{purple}.
	Moreover, we call edges that lie on the \(x\)-walk and on the \(y\)-walk \emph{green}, and set \(R = N_{(V(S),E(S) \setminus E(T))}(x)\) and \(B = N_{(V(S),E(S) \setminus E(T))}(y)\).
					
	\paragraph*{Records.}					
	We now formalize the records used in our dynamic programming procedure: a \emph{record} is given by a tuple $(\alpha_x,\alpha_y,\tau)$.
	Here $\alpha_x$ is either some vertex in $x_i$ or the drawing of an edge between $x$ and some $v \in R$ that is added to $\mathcal{T}$ and crosses a crossable edge in $\mathcal{T}$ at most once.
	$\alpha_y$ is defined analogously, i.e.\ either is some \(y_i\) or the drawing of an edge between $y$ and some $v \in B$ that is added to $\mathcal{T}$ and crosses a crossable edge in $\mathcal{T}$ at most once.
	Note that we can combinatorially represent such drawings by specifying its crossing point and endpoint, since two drawings of edges with the same crossing point and endpoint have the same topological properties and hence are equivalent for the purposes of this section.
	$\tau$ is then an auxiliary element that simplifies the description of the algorithm and will be assigned one out of $10$ values (this will be detailed in the next subsection).
	There are also two special records, \textsc{Start} and \textsc{End}.
	From their definition, one can see that the total number of records is upper-bounded by $\bigoh(|V(T)|^4)$.
					
	During our dynamic program, we will keep two sets of records: the set $\texttt{Reach}$ of records that are ``reachable'', and the set $\texttt{Proc}$ of records that have already been exhaustively ``processed''.
	Initially $\texttt{Proc}=\emptyset$ and $\texttt{Reach}=\{\texttt{Start}\}$.
	As soon as \texttt{End} is added to \texttt{Reach}, then the algorithm ascertains that $(S,T,\mathcal{T})$ is a YES-instance of 1-planar drawing extension, and we can even find a restriction that conforms to the restrictions imposed by the branching that led to the subinstance \((S,T,\mathcal{T})\).
	Such an extension can be computed via standard backtracking along the successful run of the dynamic algorithm.
	If however, at any stage $\texttt{Proc}=\texttt{Reach}$ without \texttt{End} having been added to \texttt{Reach} then we let the algorithm output ``NO''.
					
	\paragraph*{Record Types and Delimiters.}
	The type of a record (stored in $\tau$) intuitively represents which out of $10$ ``cases'' the record encodes. Here, we provide a formal description of each such case.
	We note that for each case, it will be easy to verify whether it is compatible with a certain choice of $\alpha_x$ and $\alpha_y$---if it is not, we do not consider such a ``malformed'' combination $(\alpha_x,\alpha_y,\tau)$ in our branching and for our records.
					
	We also define a \emph{delimiter} $D$ for each record type---this is a simple curve from $x$ to $y$ that separates the instance into two subinstances.
	Intuitively, in a ``correct'' branch we will later be able to assume that 
	\begin{enumerate}
		\item the drawing of every added edge in the hypothetical solution of \((S,T,\mathcal{T})\) does not intersect \(D\), i.e.\ lies completely on one side of \(D\), and
		\item we have an assignment $\lambda$ that maps all cells, the last vertex on the \(x\)-walk and the \(y\)-walk of which occurs before \(D\) on the traversals of the \(x\)-walk and the \(y\) walk, to $\{x,y,\emptyset\}$.
	\end{enumerate}
	In this sense the delimiter indicates, at which stage of the dynamic program we are; the cells the last vertex on the \(x\)-walk and the \(y\)-walk of which occurs before \(D\) on the traversals of the \(x\)-walk and the \(y\)-walk are already sufficiently processed to apply Lemma~\ref{lem:flow}, and we still need to achieve this for the remaining cells.
					
	We now describe the five types (possible entries for \(\tau\)) of our records.
	Note that all descriptions are given from the ``perspective'' of $x$, and symmetric cases also exist for $y$.
	}
\iflong{
	\begin{figure}
		\begin{subfigure}{.31\textwidth}
			\includegraphics[page=2]{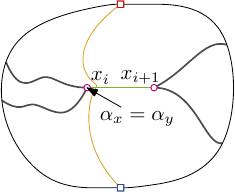}
			\caption{\textbf{Green Pointer}.}
			\label{fig:records:case1}
		\end{subfigure}
		\hfill
		\begin{subfigure}{.31\textwidth}
			\includegraphics[page=3]{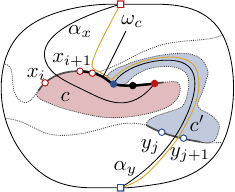}
			\caption{\textbf{Double Incursion}.}
			\label{fig:records:case2}
		\end{subfigure}	
		\hfill
		\begin{subfigure}{.31\textwidth}
			\includegraphics[page=4]{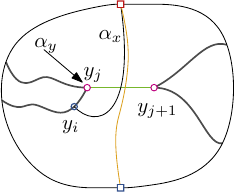}
			\caption{\textbf{Left Incursion}.}
			\label{fig:records:case3}
		\end{subfigure}
		
		\begin{subfigure}{.31\textwidth}
			\includegraphics[page=5]{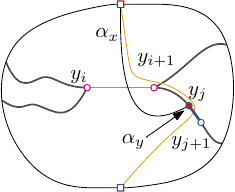}
			\caption{\textbf{Right Incursion}.}
			\label{fig:records:case4}
		\end{subfigure}	
		\begin{subfigure}{.31\textwidth}
			\includegraphics[page=6]{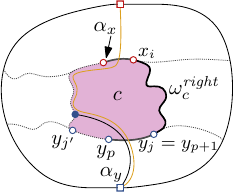}
			\caption{\textbf{Slice}.}
			\label{fig:records:case5}
		\end{subfigure}
		
		\caption{The five types of records for the dynamic program. The orange curves are the delimiters computed in each case. $ x $ and $ y $ are the red and blue rectangular vertices respectively\ifshort{.}\iflong{, colored disks are the sets $ B $ and $ R $, while white disks are the vertices $ x_i $ and $ y_j $ depending on their border-color.}}
		\label{fig:records}
	\end{figure}
	}
	
	\iflong{
	\noindent \textbf{1. Green Pointer.} $\alpha_x=\alpha_y$ are the same vertex $x_i$, and that vertex is incident to a green edge $x_ix_{i+1}$.
					
	\smallskip
	\noindent \emph{Delimiter:} The delimiter in this case is a simple curve that starts in $x$, crosses through $x_ix_{i+1}$, and then proceeds to $y$ without crossing any edges other than \(x_ix_{i+1}\), i.e.\ the delimiter remains within the starting cells of \(x\) and \(y\) with \(x_ix_{i+1}\) on their boundaries.
					
	\noindent \emph{Intuition:} This is the simplest type of our records, and is only used in a few boundary cases to signify that ``nothing noteworthy happened in this dynamic programming step'' in the sense that the cells that were created in this dynamic programming step are assumed to remain empty in the targeted hypothetical drawing extension.
					
	\smallskip
	\noindent \textbf{2. Double Incursion.} $\alpha_x$ is an edge which crosses a crossable edge $x_ix_{i+1}$ into some cell $c$ (which may, but need not, be $c_y$);
	let $\omega_c$ be the walk around the boundary of $c$ starting in \(x_{i + 1}\) and ending in the endpoint of $\alpha_x$ that does not contain \(x_i\).
	$\alpha_y$ is an edge which ends on $\omega_c$ and crosses into a cell $c'$ (which may, but need not, be a starting cell $c_x$ of \(x\)) through some crossable edge $y_jy_{j+1}$.
					
	\noindent \emph{Delimiter:} The delimiter in this case is a simple curve that starts in $y$, crosses through $y_jy_{j+1}$ (intuitively, this occurs ``behind'' the crossing point of $\alpha_y$) without crossing any edges before that, and then follows $\alpha_y$ (without crossing it) until reaching the endpoint of $\alpha_y$.
	At that point, it follows $\omega_c$ in reverse direction (without crossing it) towards $x_i$. At some point, the delimiter must reach an edge that is on the \(x\)-walk (since $x_ix_{i+1}$ is such an edge; that being said, the delimiter might already be on the \(x\)-walk when it reaches the endpoint of $\alpha_y$).
	Whenever that happens, the delimiter diverges from the \(x\)-walk and connects to $x$ without crossing any further edges, i.e.\ within the starting cell at whose boundary the delimiter first hit the \(x\)-walk.
	Notice that even if the delimiter crosses through several purple cells, only at most one of these will remain accessible from both \(x\) and \(y\) when the drawings of \(\alpha_x\) and \(\alpha_y\) are added to \(\mathcal{T}\).
					
	\noindent \emph{Intuition:} This is the most complicated record, since it covers a number of cases of non-trivial interactions between $\alpha_x$ and $\alpha_y$.
	When using it, we will ensure that our edges \(\alpha_x\) and \(\alpha_y\) have a maximality property in a targeted hypothetical extension which ensures that no added edge which is incident to $y$ ends between $x_{i+1}$ and the endpoint of $\alpha_y$, and no added edge which is incident to $x$ ends between the endpoint of $\alpha_x$ and $y_jy_{j+1}$.
	The edge $\alpha_y$ will split whatever cell it crosses into one part that remains accessible from $y$, and one part that is no longer accessible from $y$.
	$\alpha_x$ induces a similar splitting of the cell it crosses into.
					
	\smallskip
	\noindent \textbf{3. Left Incursion.} $\alpha_x$ is an edge $e$ which ends in $y_i$ after crossing across a crossable green edge $y_{j}y_{j+1}$, where $j>i$ and $\alpha_y=y_{j}$.
					
	\noindent \emph{Delimiter:} The delimiter in this case is a simple curve that starts in $x$, crosses through $y_jy_{j+1}$ (intuitively, this occurs ``behind'' the crossing point of $e$), and then proceeds to $y$ without crossing any edges other than \(y_iy_{i+1}\), i.e.\ the delimiter remains within the starting cells of \(x\) and \(y\) with \(y_iy_{i+1}\) on their boundaries.
					
	\noindent \emph{Intuition:} This record type represents the case where only one ``half'' of a \textbf{Double Incursion} is present.
	$e$ splits a cell off $c_y$ that is no longer accessible from \(y\).
					
	\smallskip
	\noindent \textbf{4. Right Incursion.} $\alpha_x$ is an edge $e$ which crosses a crossable edge $y_iy_{i+1}$ and ends in $y_j$ (where $j>i$) and the previous cases \textbf{Double Incursion} and \textbf{Left Incursion} do not apply. $\alpha_y=y_j$.
					
	\noindent \emph{Delimiter:} The delimiter in this case is a simple curve that starts in $y$, crosses through $y_jy_{j+1}$ (intuitively, this occurs ``behind'' $y_j$) without crossing any other edges before that, and then proceeds along the \(y\)-walk in reverse direction towards $y_i$ without crossing it.
	Since $y_iy_{i+1}$ is a green edge, the delimiter eventually reaches the \(x\)-walk (it might, in fact, already be there).
	Whenever that happens, the delimiter diverges from the \(x\)-walk and connects to $x$ without crossing any further edges, i.e.\ within the starting cell at whose boundary the delimiter first hit the \(x\)-walk.
					
	\noindent \emph{Intuition:} This record represents the case where only one ``half'' of a \textbf{Double Incursion} is present, and can be viewed as a degenerate \textbf{Double Incursion} where $\alpha_y$ is collapsed into a single vertex.
					
	\smallskip
	\noindent \textbf{5. Slice.} $\alpha_x$ is either a vertex on the boundary of a purple cell $c$, or an edge which crosses a crossable edge of \(T\) into such a cell $c$.
	$\alpha_y$ is either a vertex on the boundary of $c$ or an edge which crosses a crossable of \(T\) into $c$.
	Moreover, if $\omega^\emph{right}_c$ is the unique walk along the boundary of \(c\) from the maximum-index vertex $x_i$ of the \(x\)-walk on the boundary of $c$ to the maximum-index vertex $y_j$ of the \(y\)-walk on the boundary of $c$ that does not contain the lowest-index $y_{j'}$ on the boundary of $c$, then neither $\alpha_y$ nor $\alpha_x$ end on (if they are edges) or are in (if they are vertices) $\omega^\emph{right}_c$.
					
	\noindent \emph{Delimiter:}
	Note that the boundary of $c$ might intersect a ``previous delimiter'' (later denoted by \(D'\)), which we then treat as part of the boundary of \(c\).
	The delimiter starts in $y$.
	If $\alpha_y$ is an edge that crosses $y_{p}y_{p+1}$, the delimiter then crosses through $y_{p}y_{p+1}$ without crossing any other edges before that and continues along \(\alpha_y\) and then along the part of the boundary of $c$ that does not contain \(\omega^\emph{right}_c\) on the side of \(c\).
	If \(\alpha_y\) is a vertex $y_p$, the delimiter then crosses through $y_{p}y_{p+1}$ without crossing any edges before that and immediately continues along the part of the boundary of $c$ that does not contain \(\omega^\emph{right}_c\) on the side of \(c\).
	Once the delimiter reaches an endpoint or crossing point of $\alpha_x$ in case \(\alpha_x\) is an edge (whichever comes first); or \(\alpha_x\) itself otherwise, it proceeds analogously as for $\alpha_y$ to reach $x$.
	
	\noindent \emph{Intuition:} This record represents, in a unified way, a multitude of situations that may occur inside a purple cell.
	It will be used in a context where $\alpha_x$ and $\alpha_y$ satisfy a certain maximality condition which ensures that, in the part up to the delimiter (i.e., between the previous and current delimiter), edges from $x$ only enter $c$ ``under'' $\alpha_y$, and similarly edges from $x$ only enter $c$ ``under'' $\alpha_x$.
					
	\paragraph*{Branching Steps.}
	We can now describe the actual dynamic program $\mathbb{A}$.
	The top-level procedure iterated by $\mathbb{A}$ is that it picks some record from $\textsc{Reach}\setminus\textsc{Proc}$ (at the beginning this is \texttt{Start}, but later it will be some $(\alpha'_x,\alpha'_y,\tau')$).
	For records with record types, i.e.\ records that are not \texttt{Start} or \texttt{End}, it then computes a delimiter \(D'\).
					
	Now, $ \mathbb{A} $ exhaustively branches over all possible choices of $\alpha_x$, $\alpha_y$, and $\tau$, and checks that the record $(\alpha_x, \alpha_y,\tau)$ satisfies the conditions for the given type $\tau$ as described above (i.e., that it is not malformed).
	It then adds $\alpha_x$, $\alpha_y$ into the drawing, computes the delimiter $D$ for $(\alpha_x, \alpha_y,\tau)$, as described above in the corresponding paragraph for \(\tau\).
	It then checks that $D$ occurs ``after'' $D'$ with respect to the traversals of both the \(x\)-walk and the \(y\)-walk.
	If either of these conditions is not satisfied, it discards this choice.
					
	If all of the above checks were successful, $\mathbb{A}$ constructs a mapping $\lambda$ to verify whether the branch represents a valid step of the algorithm---notably, whether the added edges whose endpoints are enclosed by $D'$ and $D$, and potentially \(\alpha_x\) and \(\alpha_y\), can be drawn into the partial drawing in a 1-planar way without crossing \(D'\) or \(D'\).
	We can already partially define \(\lambda\), by setting \(\lambda(c)\) to \(x\) for red cells \(c\), to \(y\) for blue cells \(c\),
	and to \(\emptyset\) for occupied cells \(c\).
	We consider the cells of the drawing resulting by inserting \(D'\) and \(D\) into \(\mathcal{T}\) which are enclosed by \(D'\) and \(D\).
	We call these cells \emph{delimited}.
	Delimited cells which are (not necessarily proper) subcells of cells of \(\mathcal{T}\) for which \(\lambda\) is already defined, inherit the same \(\lambda\)-value.
	We call cells for which \(\lambda\) is not yet cell defined \emph{unassigned}.
	
	In the special case where \(D'\) is a \textbf{Slice}, and \(D\) is not a \textbf{Slice} involving the same purple cell \(c'\),
	we perform an additional branching subroutine to correctly assign the unassigned delimited part of the purple cell intersected by \(D'\).
	For this we branch over every drawing of a potential undominated added edge $\beta_x$ from $x$ that crosses into $c'$ and ends in $\omega^\emph{right}_{c'}$, where $\omega^\emph{right}_{c'}$ is defined as in the description of the \textbf{Slice} record type (also considering the case that no such edge exists).
	Similarly, we branch over every drawing of a potential undominated edge $\beta_y$ from $y$ that crosses into $c'$ and ends in $\omega^\emph{right}_{c'}$.
	We then use these drawings to split $c'$ into up to three subcells, and have $\lambda$ map these as follows:
	the unique part delimited by $\beta_x$ that contains the maximum-index $x_i$ in $c'$ is mapped to $x$,
	the unique part delimited by $\beta_y$ that contains the maximum-index $y_j$ in $c'$ is mapped to $y$,
	and the rest of $c'$ is mapped to \(\emptyset\).
	
	For all further unassigned cells we extend \(\lambda\) depending on the record type of the newly guessed record (assuming the cases are given w.r.t.\ $x$, as in the description of the cases; mirrored cases are treated analogously, but with swapped $x$ and $y$):
					
	\begin{enumerate}
		\item If $\tau=\textbf{Green Pointer}$, set $\lambda(c)=\emptyset$ for every unassigned cell;
		\item If $\tau=\textbf{Double Incursion}$, $\alpha_x$ splits the cell \(c\) of \(\mathcal{T}\) it crosses into into two cells.
		Let $x_ix_{i+1}$ be the edge of \(T\) that $\alpha_x$ crosses and \(c\) be the cell of \(\mathcal{T}\) that \(\alpha_x\) crosses into.
		Then \(c\) is subdivided into two cells by \(\alpha_x\), exactly one of which has \(x_{i+1}\) on its boundary.
		We map this subcell of \(c\) to \(x\) while the other subcell of \(c\) is mapped to\(y\).
		Similarly $\alpha_y$ crosses an edge \(y_jy_{j+1}\) of \(T\) and splits the cell $c'$ of \(\mathcal{T}\) it crosses into into two subcells, exactly one of which has \(y_j\) on its boundary.
		We map this subcell of \(c'\) to \(y\) while the other one is mapped to \(x\).
		All remaining unassigned cells are mapped to \(\emptyset\).
		\item If $\tau=\textbf{Left Incursion}$ or $\tau=\textbf{Right Incursion}$, $\alpha_x$ splits $c_y$ into two subcells, where the part containing $y$ is mapped to $y$ and the other part to $x$.
		All remaining unassigned cells are mapped to \(\emptyset\).
		\item If $\tau=\textbf{Slice}$, and $\alpha_x$ is an edge, it splits a purple cell $c$ of \(\mathcal{T}\) into two subcells, exactly one of which does not contain \(\omega^\emph{right}_c\) as defined in the definition of the \textbf{Slice} record type on its boundary.
		We map this subcell to \(x\), if it exists.
		Similarly if \(\alpha_y\) is an edge, it splits a purple cell $c$ of \(\mathcal{T}\) into two subcells, exactly one of which does not contain \(\omega^\emph{right}_c\) as defined in the definition of the \textbf{Slice} record type on its boundary, and we map this subcell to \(y\), if it exists.
		All remaining unassigned cells are mapped to \(\emptyset\).
	\end{enumerate}
					
	Finally, with this mapping, we invoke Lemma~\ref{lem:flow} on the (sub-)instance that arises by considering the instance restricted to delimited cells.
	If the resulting network flow instance is successful, we add $(\alpha_x,\alpha_y,\tau)$ to \textsc{Reach}.
	Once all branches have been exhausted for a given record, we add it to \textsc{Proc}.
					
	We also include the special record \texttt{End} among the records we branch over.
	For this record, we do not include a delimiter \(D\), i.e.\
	we consider all cells that occur after \(D'\) along both the traversal of the \(x\)-walk and the \(y\)-walk.
	If $\mathbb{A}$ adds \texttt{End} to \texttt{Reach}, then \textsc{A} outputs ``YES'' (and backtracks to find a suitable drawing by inductively invoking Lemma~\ref{lem:flow} on the relevant (sub-)instances of $(S,T,\mathcal{T})$).
	If however at some stage \texttt{Proc}=\texttt{Reach} then $\mathbb{A}$ rejects $(S,T,\mathcal{T})$.
	The total running time of $\mathbb{A}$ is upper bounded by $\bigoh(|V(T)|^{12})$, since the number of records is at most $\bigoh(|V(T)|^4)$ and the branching subroutines take time at most $\bigoh(|V(T)|^8)$.
					
	\paragraph*{Domination and Correctness.}
	All that remains now is to show that $\mathbb{A}$ is, in fact, correct. To do so, we will need the notion of \emph{undominated edges}, which formalizes the maximality assumption that we informally hinted at in the intuitive description of our record types.
	Given a hypothetical solution $\mathcal{S}$ to $(S,T,\mathcal{T})$, an edge $e$ starting in $x$ is \emph{dominated} if one of the following holds:
	\begin{itemize}
		\item $e$ crosses into $c_y$, and the part of \(c_y\) it encloses with the boundary of $c_y$ that does not contain $y$ is contained in the part of \(c_y\) enclosed by another added edge from $x$ and the boundary of $c_y$ not containing $y$;
		\item $e$ crosses into a purple cell $c$ of \(\mathcal{T}\), and one of the two parts of \(c\) it encloses together with the boundary of \(c\) is contained in one part of \(c\) enclosed by another added edge from $x$ and the boundary of $c_y$.
	\end{itemize}
					
	An edge is \emph{undominated} if it is not dominated.
	The definition for added edges starting in \(y\) is analogous.
	We extend the notion of domination towards vertices by considering a vertex to be a loop of arbitrarily small size (hence, every edge ``enclosing'' that vertex dominates it).
	The role of undominated edges (and vertices) is that they will help us identify which parts of a hypothetical solution $\mathcal{S}$ to focus on in order to find a relevant step of $\mathbb{A}$ for our correctness proof.
					
	\begin{lemma}
		\label{lem:correct}
		$\mathbb{A}$ is correct.
	\end{lemma}

	\begin{proof}
		Assume $\mathbb{A}$ outputs ``YES''. Then one can follow the sequence of branching and records considered by $\mathbb{A}$ that leads to $\texttt{End}$ being added to \textsc{Reach}.
		During our branching we make sure to only consider records whose edges do not cross uncrossable edges or the delimiter of the previous record, and so an iterative application of Lemma~\ref{lem:flow} on each of our branching steps guarantees that if we follow this sequence we will end up routing all added edges from $x$, $y$ to their required endpoints---i.e., with a solution to $(S,T,\mathcal{T})$.
						
		Conversely, assume $(S,T,\mathcal{T})$ has a solution $\mathcal{S}$.
		To complete the proof, we will inductively build a sequence of records whose delimiters are not crossed by any edge in $\mathcal{S}$.
		This property is trivially true when we have no delimiter, i.e.\ for our initial record \texttt{Start}, but in general we will simply assume we have found an arbitrary delimiter $D'$ that has this property and added it to \texttt{Reach}.
		It now suffices to show that $\mathbb{A}$ will be able to find (via branching) and add an additional record to its set \texttt{Reach} that strictly pushes the delimiter further with respect to the traversals of the \(x\)-walk or the \(y-walk\); since the total number of records is bounded, this would necessarily imply that at some point $\mathbb{A}$ would terminate with a positive output by adding \texttt{End} to \texttt{Reach}.
						
		We will now perform an exhaustive case distinction, where for each case we will identify a record that would be detected by branching and where the existence of $\mathcal{S}$ would guarantee that the network flow instance constructed by Lemma~\ref{lem:flow} using the corresponding $D$ and $\lambda$ (constructed as per the description of our branching steps) terminates successfully. 
		
		Let us begin by considering the case where $D'$ does not cross through any purple cell.
		Let $x_i$ be the first vertex to occur strictly after $D'$ along the \(x\)-walk that is:
		\begin{enumerate}[\text{Case} 1]
			\item the endpoint of an added edge $e$ starting in $y$ that crosses the \(x\)-walk into a starting cell \(c_x\) of \(x\), or
			\item the endpoint of an added edge $e$ starting in $x$ that crosses the \(y\)-walk into a starting cell \(c_y\) of \(y\), or
			\item incident to a green edge $e=x_{i-1}x_i$ crossed by an added edge starting in $x$, or
			\item incident to a green edge $e=x_{i-1}x_i$ crossed by an added edge starting in $y$, or
			\item a vertex on the boundary of a purple cell $c$, or
			\item no such vertex exists.
		\end{enumerate}
						
		In Cases 1 and 2, this gives rise to a \textbf{Left Incursion}. In both cases, it is easy to verify that an appropriate restriction of $\mathcal{S}$ is a $\lambda$-consistent solution and hence the branching would detect the corresponding record and add it to \texttt{Reach}, as required.
						
		In Case 3. and 4., this gives rise to a \textbf{Right Incursion} or \textbf{Double Incursion}, depending on whether there is an edge $e'$ from $y$ (or $x$, respectively) ending on the walk between the crossing point and the endpoint of $e$. If no such edge exists, we obtain a \textbf{Right Incursion}. Moreover, it is once again easy to verify that an appropriate restriction of $\mathcal{S}$ is $\lambda$-consistent:
		Let \(c_x\) and \(c_y\) be the starting cells of \(x\) and \(y\) respectively that have \(e\) on their shared boundary.
		The ``inside'' of the added edge that crosses \(e\) within $c_y$ is not reachable from $y$ directly anymore (as it would require crossing that added edge), and also not via the boundary with $c_x$ due to the non-existence of $e'$.
		Hence the new delimiter will not be crossed by $\mathcal{S}$ and the restriction of $\mathcal{S}$ to the cells delimited by \(D'\) and the corresponding \(D\) will be $\lambda$-consistent.
						
		On the other hand, if an edge $e'$ as described above exists, let us fix $e'$ to be the unique undominated edge with this property.
		Then the ``inside'' of $e'$ within the face it crosses into is only reachable from $y$, while the ``outside'' of $e'$ cannot be reached by $y$ anymore (due to $e'$ being undominated).
		The argument for $e$ is the same as above, and putting these two together we obtain that the appropriate restriction of $\mathcal{S}$ is also $\lambda$-consistent and $D$ is not crossed.
						
		If Case 5. applies to $x_i$ as well as to the symmetrically constructed $y_j$, then by construction \(x_i\) and \(y_j\) lie on the boundary of the same purple cell \(c\), and the pair forms a (degenerate) \textbf{Slice} in $c$.
		Let $e$ be the unique undominated edge dominating $x_i$ (if none exists, we use $x_i$ instead),
		and the same for $e'$ and $y_j$.
		In this case, the record $(e,e',\textbf{Slice})$ would pass our branching test.
		If Case 5. applies to $x_i$, but not to the symmetrically constructed $y_j$, then depending on which of the Cases 1. - 4. applies to \(y_j\) symmetric considerations can be made from the viewpoint of \(y\).
		We are left with the case that Case 5. applies to \(x_i\) and Case 6. applies to the symmetrically constructed \(y_j\).
		However, this case cannot occur, as then \(D'\) would occur after the maximum-index vertex on the shared boundary of \(c\) and \(c_y\) along the \(y\)-walk, but before the maximum-index vertex on the intersection of the boundary of \(c\) and the \(x\)-walk.
		This never happens in the construction of our delimiters.
		
		Finally, in Case 6., if any of the Cases 1.-5. applies to the symmetrically constructed \(y_j\) symmetric considerations can be made from the viewpoint of \(y\).
		We are left with the case that Case 6. applies to both \(x_i\) and the symmetrically constructed \(y_j\).
		Then the added edges already lie entirely before \(D'\) in \(\mathcal{S}\), and thus \texttt{End} would pass our branching test, which means $\mathbb(A)$ would correctly identify a solution to $(S,T,\mathcal{T})$.
						
		Now, consider the case where $D'$ does cross through a purple cell $c$.
		Let $x_i$ be the minimum-index vertex to occur strictly after $D'$ along the \(x\)-walk that is:
		\begin{enumerate}[\text{Case} A]
			\item the endpoint of an edge that crosses into $c$ or incident to an edge that is crossed by an added edge to enter \(c\), or
			\item incident to the green edge $x_ix_{i+1}$, or
			\item a vertex on the boundary of a different purple cell $c'$, or
			\item no such vertex exists.
		\end{enumerate}
						
		Case C is handled analogously as Case 5. in the previous case distinction, and Cases B and D (which is itself analogous to Case 6. from before) are also simple.
		The same holds, when the symmetrically defined vertex \(y_j\) on the \(y\)-walk satisfies one of these cases.
		The by far most complicated remaining case is Case A for both \(x_i\) and \(y_j\), which has several subcases.
		
		First, we consider what happens if both $x_ix_{i+1}$ and $y_jy_{j+1}$ are crossed by the undominated edges $e,e'$ (respectively) which end in $\omega_c^\emph{right}$.
		Then, let us consider which of Case 2., Case 3., and Case 4. occurs after $c$ (along the traversal of the \(x\)-walk and the \(y\)-walk): we will use these to define our new record, say $(\alpha_x,\alpha_y,\tau)$.
		For this record, it is just as easy to argue that the corresponding delimiter is not crossed by $\mathcal{S}$ as before.
		As for $\lambda$, the basic mapping associated with $(\alpha_x,\alpha_y,\tau)$ would leave the delimited cell $c$ unassigned---that is where the additional branching rule via $\beta_x$ and $\beta_y$ comes into play.
		In the case where $\beta_x=e$ and $\beta_y=e'$, the obtained branch results in a mapping $\lambda$ such that the appropriate restriction of $\mathcal{S}$ is $\lambda$-consistent.
						
		If the previous subcase of Case A does not occur, we distinguish the special case where an edge $e$ from $x$ crosses into $c$ via $x_ix_{i+1}$ and an edge $e'$ crosses into a cell \emph{different} from $c$ and ends on the walk from $x_{i+1}$ to the endpoint of $e$ along the boundary of \(c\) that does not contain \(x_i\)---in this case, the edges $e$ and $e'$ represent a \textbf{Double Incursion}.
		Crucially, since $e$ does not end on $\omega_c^\emph{right}$ (as that was handled by the previous subcase), $e'$ must cross into the a starting cell $c_x$ of \(x\).
		As before, we will assume that $e'$ is undominated. The resulting $\lambda$ then splits $c_x$ and $c$ in a way that is forced by $e$ (up to $D$, no added edge incident to \(x\) can enter in $c$ outside of the area delimited by $e$) and $e'$ (up to $D$, no added edge incident to \(y\) can enter in $c_x$ outside of the area delimited by $e'$).
						
		Finally, if none of the previous subcases of Case A occurs, then we are left with a situation that will boil down to a \textbf{Slice}---the last obstacle remaining is to identify the relevant edges. In particular, let $e$ be the endpoint of the edge that crosses near or ends at $x_i$, and let $e'$ be the analogously defined edge for $y_j$.
		If both edges have the same endpoint among \(\{x,y\}\), then it is trivial to see that they must be the same edge, and that will be the (single) edge used to define our \textbf{Slice}.
		Otherwise, w.l.o.g.\ $e$ starts in $x$ and $e'$ in $y$.
		We perform the following check: does the part of $c$ split off by $e$ which does not contain $\omega_c^\emph{right}$ on its boundary contain a part of $c$ split off by $e'$?
		If this is not the case in either direction, then the new \textbf{Slice} is defined by both edges i.e.\ the record \((e,e',\textbf{Slice})\);
		otherwise, it is defined by the single edge that delimits the smaller part of $c$ not containing $\omega_c^\emph{right}$ and the vertex \(x_i\) or \(y_j\) which is not related to that edges endpoint, i.e.\ the record \((e,y_j,\textbf{Slice})\) or \((x_i,e',\textbf{Slice})\).
		In all of these cases, it is easy to verify that $D$ indeed is not crossed by $\mathcal{S}$ and that the appropriate restriction of $\mathcal{S}$ is a $\lambda$-consistent solution.
	\end{proof}
	
	Using Lemma~\ref{lem:correct} and the running time analysis from above yields:
	\begin{lemma}\label{thm:2vertices}
	\subproblem\ can be solved in $\bigoh(|V(T)|^{12})$.
	\end{lemma}
	}

We are now ready to formally prove our main theorem.

\iflong{
\begin{proof}[Proof of Theorem~\ref{thm:main}]
	By Corollary~\ref{cor:reduction}, we can reduce the instance $(G,H,\cH)$ in $n^{\bigoh(\kappa^3)}$ time into $n^{\bigoh(\kappa^3)}$ many instance{s} of \textsc{Untangled $\kappa$-Bounded 1-Planar Drawing Extension} such that $(G,H,\cH)$ is yes-instance if and only if at least one of these instance{s} is yes-instance. Now let $(G,H', \cH')$ be one of these instances. By Lemma~\ref{lem:baseregions}, there are at most $\bigoh(\kappa^3))$ base regions in any \nice\ extension of $\cH'$. Each base region is defined by two added edges incident to the same vertex and their intersections with the boundary of the cell incident with the vertex. This allows us to enumerate all $n^{\bigoh(k^3)}$ possibilities for the base regions. For each possibility, we can now include the guessed edges that bound the regions which now creates $\bigoh(\kappa^3)$ base cell, such that we can assume that each edge starts in one of the base cells. We let the resulting graph and its drawing be $H''$ and $\cH''$, respectively. By Proposition~\ref{prop:2markbound}, the number of vertices $v$ in $V(G)\setminus \aV$ such that the new edges incident to $v$ start in more than two different base regions in any \vnice\ extension of $\cH''$ is bounded by $\bigoh(\kappa^3)$. We can again branch into $n^{\bigoh(\kappa^4)}$ many branches. In each branch, we guess the subset of vertices that are incident to edges starting in more than two base region together with drawing of all at most $\kappa$ edges incident to each of the vertices in this set.  By Lemma~\ref{lem:stretchbound} there are at most $\bigoh(\kappa^3)$ cells that are accessible by three base cells. By Lemma~\ref{lem:stretchBranching}, for each such cell $c$ there is a set of at most $\bigoh(\kappa^6)$ new edges such that if we draw these edges, then each subcell of $c$ intersects only new edges that starts in two different base cells. By guessing the drawings of these edges for each of $\bigoh(\kappa^3)$ many cells and guessing of the assignment of at most two accessible base cells for each newly created cell, we branch further into $n^{\bigoh(\kappa^9)}$ many instances. In each we instance, every cell $c$ contains a set of two markers (base cells) \(\nu(c)\) and we are only looking for a solution such that if a drawing of a new edge $e$ intersects $c$, then $e$ starts in one of the base cells in \(\nu(c)\).
	
	To solve this case, we apply the Turing reduction given by Lemma~\ref{lem:reductiontoTwoVertices}, which results in $n^{\bigoh(\kappa^{28})}$ branches and solve the resulting instances of \subproblem\ using Lemma~\ref{thm:2vertices}. The total number of instances we need to consider and solve is 
	\[
	n^{\bigoh(\kappa^3)}\cdot n^{\bigoh(\kappa^3)}\cdot n^{\bigoh(\kappa^4)}\cdot n^{\bigoh(\kappa^9)}\cdot n^{\bigoh(\kappa^{28})} = n^{\bigoh(\kappa^{28})}\qedhere
	\]

\end{proof}
}

\section{Concluding Remarks}

In this paper we have presented a constructive polynomial-time algorithm for extending partial 1-planar graph drawings (or report that no such extension exists) in the restricted case that the edge+vertex deletion distance $\kappa$ between the partial drawing and its extension is bounded.\iflong{
This closes one of the main unresolved questions identified in previous work that investigated the extension of 1-planar drawings~\cite{Aug1P-ICALP}.}
We believe that the most interesting open question in this direction is whether the 1-planar drawing extension problem is fixed-parameter tractable not only w.r.t.\ the edge deletion distance~\cite{Aug1P-ICALP}, but also w.r.t.\ the edge+vertex deletion distance. \ifshort{It would also be worthwhile to obtain an understanding of the complexity of the problem when restricted to instances with natural structural properties, such as having bounded treewidth.}
\iflong{One may also investigate the complexity of this as well as other prominent extension problems through the lens of structural restrictions on the input graph---for instance, what would happen if we restrict the set of instances to those where $G$ has bounded treewidth?}

Last but not least, we note that while the requirement on the connectivity of $H$ is well motivated from an application perspective and has also been used in other drawing extension settings~\cite{mels-lam-95,mnr-ecpdg-15,HongN08,Aug1P-ICALP}, the problem of course remains of interest when this requirement is dropped. 
Our techniques and results do not immediately carry over to the case where $H$ is disconnected, and generalizing the presented results in this direction is an 
interesting question left for future research.

\bibliographystyle{plainurl}
\bibliography{references}
	
\end{document}